\documentclass[11pt,reqno]{amsart}
\title[Self-distributivity, braces $\&$  the Yang-Baxter equation]
{Self-distributive structures, braces $\&$\\ the Yang-Baxter equation}
\author[Anastasia Doikou]{Anastasia Doikou}

\address[Anastasia Doikou] {Department of Mathematics, Heriot-Watt University,
Edinburgh EH14 4AS $\&$ Maxwell Institute for Mathematical Sciences, Edinburgh EH8 9BT, UK}
\email{a.doikou@hw.ac.uk}

  \usepackage{latexsym} 
  \usepackage[all]{xy}
  \usepackage{amsfonts} 
  \usepackage{amsthm} 
  \usepackage{amsmath} 
  \usepackage{amssymb}
  \usepackage{pifont}  
  \usepackage{enumerate}
  \usepackage{dcolumn}
  \usepackage{comment}
  \usepackage{hyperref}  
  \usepackage{tikz-cd}
\usepackage[english]{babel}
  \usepackage{amsfonts} 
\usepackage[utf8]{inputenc}
\usepackage{subcaption}
\usepackage[colorinlistoftodos]{todonotes}
\usepackage{indentfirst}
\usepackage{pifont}  
  \usepackage{enumerate}
  \usepackage{dcolumn}
  \usepackage{comment}
\usepackage[capitalise]{cleveref}
\usepackage{multirow}
  \usepackage{graphicx}
\graphicspath{ {./images/} }

 \usepackage{tikz}
\usetikzlibrary{arrows}

 \newcolumntype{2}{D{.}{}{2.0}}
  \xyoption{2cell}

\newcommand{\hiddenpower}[2] { \ifnum \numexpr#2=1 #1 \else #1^#2 \fi }
\numberwithin{equation}{section}

\def\be{\begin{equation}}
\def\ee{\end{equation}}
\def\ba{\begin{eqnarray}}
\def\ea{\end{eqnarray}}

\newcommand{\cal}{\mathcal}

\usepackage{xparse}
\ExplSyntaxOn
\newcounter{diff_order}
\newcounter{diff_power}

\newcommand{\rawdiff}[3]
{
	\setcounter{diff_order}{0}
	\clist_map_inline:nn{#3}{\stepcounter{diff_order}}
	
	\frac{\hiddenpower{#1}{\thediff_order} #2}
	{
		\def\old_var{DefaultValue}
		\setcounter{diff_power}{0}
		
		\clist_map_inline:nn{#3}
		{
			\def\new_var{##1}
			\ifnum \thediff_power=0
				\stepcounter{diff_power}
			\else
				\tl_if_eq:NNTF \new_var \old_var
				{\stepcounter{diff_power}}
				{
					#1 \hiddenpower{\old_var}{\thediff_power}
					\setcounter{diff_power}{1}
				}
			\fi

			\def\old_var{##1}
		}
		
		#1 \hiddenpower{\old_var}{\thediff_power}
	}
}
\def\Label#1{\label{#1}\ifmmode\llap{[#1] }\else 
  \marginpar{\smash{\hbox{\tiny [#1]}}}\fi} 
  \def\Label{\label} 

\ExplSyntaxOff

\newlength{\bibitemsep}
\newlength{\bibparskip}
\let\oldthebibliography\thebibliography
\renewcommand\thebibliography[1]{%
  \oldthebibliography{#1}%
  \setlength{\parskip}{\bibitemsep}%
  \setlength{\itemsep}{\bibparskip}%
}

\newtheorem{thm}{Theorem}[section]
\newtheorem{lemma}[thm]{Lemma}
\newtheorem{cor}[thm]{Corollary}
\newtheorem{pro}[thm]{Proposition}
\newtheorem{defn}[thm]{Definition}

\newtheorem{rem}[thm]{Remark}
\newtheorem{exa}[thm]{Example}

\newcommand{\id}{\operatorname{id}}

\newcommand{\blue}[1]{\textcolor{blue}{#1}}

\newenvironment{wideregion}[1][0mm]
    {\vspace{#1}\list{}{\leftmargin-2cm\rightmargin\leftmargin}\item\relax}
    {\endlist}

\newenvironment{widegather }{\wideregion[-9mm]\gather}{\endgather\endwideregion}

\begin{document}
\vskip 0.8in

\hfill
 \begin{abstract} 
The theory of the set-theoretic Yang-Baxter equation is reviewed from a purely algebraic point of view. 
We recall certain algebraic structures called shelves, racks and quandles. These objects satisfy a {\it self-distributivity} condition and lead to solutions of the Yang-Baxter equation. The quantum algebra as well as the integrability associated to Baxterized involutive set-theoretic solutions is briefly discussed. 
We then present the theory of the universal algebras associated to rack and general set-theoretic solutions. We show that these are quasi-triangular Hopf algebras and we derive the universal set-theoretic Drinfel'd twist. It is shown that this is an admissible twist allowing the derivation of the universal set-theoretic ${\cal R}$-matrix.

\end{abstract}
\maketitle

\date{}




\section{Introduction} 

\noindent 
The Yang-Baxter equation (YBE) was first introduced in a purely physical context in \cite{Yang} as the main mathematical tool for the investigation of quantum systems with many particle interactions, and in \cite{Baxter} 
for the study of statistical model known as the anisotropic Heisenberg magnet. The idea of
set-theoretic solutions to the Yang-Baxter equation was suggested in early 90's by Drinfel'd \cite{Drin} and since then, set-theoretic solutions
have been extensively investigated primarily by means of representations of the braid group, but almost exclusively for the parameter free Yang-Baxter equation (see for instance \cite{EtScSo99, GuaVen, Ru05, Ru07}).
The investigation of set-theoretic solutions of the  Yang-Baxter equation and the associated algebraic structures is a highly active research field that has been particularly prolific, given that a significant number of related studies has been produced over the past
several years (see for instance \cite{Bachi}--\cite{BrzRybMer}  \cite{CaCaSt22}--\cite{Doikoup},\cite{GatMaj}--\cite{GatMaj2}, 
\cite{Pili, JesKub, Lebed, LebVen, chin, Smo1}).   
The study of the set-theoretic Yang-Baxter equation has produced numerous significant connections to distinct mathematical
areas, such as group theory,  algebraic number theory,  Hopf-Galois extensions, non-commutative rings,  knot theory, Hopf algebras and quantum groups,  universal algebras, groupoids, heaps and trusses, pointed Hopf algebras, Yetter-Drinfel'd modules and Nichols algebras (see for instance among \cite{Andru, braceh, Bachi, Brz:Lietruss, BrzRyb:con, BrzRybMer, EtScSo99}, 
\cite{Pili}-\cite{Jesp2},\cite{Jo82}-\cite{LebVen}, \cite{Ru07,Smo1}).
 Moreover,  interesting links  between the
set-theoretic Yang-Baxter equation and geometric crystals \cite{Crystals, Crystal2}, or soliton cellular automatons \cite{Cell, Cell2} have been shown.  Concrete connections with quantum spin-chain like systems were also made in \cite{DoiSmo2, DoiSmo1}.

We note that set-theoretic solution for the parametric Yang-Baxter equation (Yang-Baxter 
maps) have been primarily studied up to date only
in the context of classical discrete and continuous integrable systems connected also to the notion of
Darboux-B\"acklund transformation or the discrete zero curvature condition in the Lax pair formulation 
and the refactorization problem and soliton interactions (see for instance \cite{Adler, Papa,  Papa2, Veselov}, \cite{Abl, Cau, Tsu}).  The refactorization is also synonymous to the so called 
Bianchi's permutability, which describes the exchange of two consecutive B\"acklund transformations. Specifically, Figure 1 describes 
graphically the construction of say  a two-soliton solution for some integrable non-linear ODE or 
PDE via two consecutive  B\"acklund transformations and the use of Bianchi's permutability.

\begin{center}
\begin{eqnarray}
& &\text{\large $ \begin{tikzcd}[row sep = small, column sep = large, ampersand replacement = \&]
	\- \& w_1^{\phantom{1}} \arrow[dr, "\lambda_2", end anchor = {north west}] \& \- \\
	w_0 \arrow[ur, "\lambda_1", start anchor = {north east}] \arrow[dr, swap, "\lambda_2", start anchor = {south east}] \& \- \& w_{12} = w_{21} \\
	\- \& w_2^{\phantom{2}} \arrow[ur, swap, "\lambda_1", end anchor = {south west}] \& \-
\end{tikzcd} $}\nonumber \\
& & \qquad \mbox{1. Bianchi Permutability} \nonumber
 \label{lattice}
\end{eqnarray}
\end{center}
The set-theoretic Yang-Baxter equation may be also seen as a cube or 3D
consistency condition in classical integrable discrete systems (discrete time evolution), see Figure 2.

$ $

\begin{center}
\includegraphics{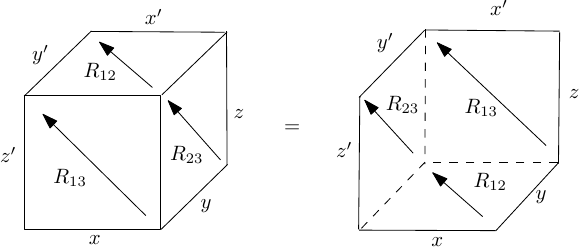}
\end{center}
\begin{center}
{ 2. The 3D consistency condition}
\end{center}
In \cite{Doikoup}  an entirely algebraic analysis for the parametric 
set-theoretic Yang-Baxter equation was undertaken and purely algebraic solutions were produced.
Earlier works on the algebraic structures as well as the associated admissible Drinfel'd twists of the non-parametric, set-theoretic Yang-Baxter equation provided a basic algebraic blueprint \cite {Doikoutw, DoGhVl, DoRySt}.

We introduce now the {\it parameter free} set-theoretic braid equation. Following \cite{Drin}, 
given a non-empty set $X$, a map $\check r:X\times X\to X\times X$ is said 
to be a \emph{set-theoretic solution of the braid equation}, if $\check r$ satisfies the \emph{braid identity}
\begin{equation}\label{eq:braid}
\left(\check r\times\id_X\right)
\left(\id_X\times\check r\right)
\left(\check r\times\id_X\right)
= 
\left(\id_X\times\check r\right)
\left(\check r\times\id_X\right)
\left(\id_X\times\check r\right).
\end{equation}
We call such a map $\check r$ simply a \emph{solution} and write $\left(X,\check r\right)$ to 
denote a solution $\check r$ on a set $X$. Besides, if we write $\check r\left(a,b\right) = \left(\sigma_a\left(b\right), 
\tau_b\left(a\right)\right)$, with $\sigma_a, \tau_a$ maps from $X$ into itself,  then $\check r$ is said to 
be \emph{left non-degenerate} if $\sigma_a$ is bijective for every $a\in X$, \emph{right non-degenerate} 
if $\tau_a$ is bijective for every $a\in X$, and \emph{non-degenerate} if $\check r$ is both left and right non-degenerate.
Furthermore, if $\check r$ is a solution such that $\check r^2=\id_{X \times X}$, then $\check r$ is said to be \emph{involutive}. \\
It is also worth recalling the connection between the set-theoretic braid equation (\ref{eq:braid}) and the 
set-theoretic Yang-Baxter equation (see also e.g. \cite{Jimbo1, Jimbo2}). 
We  introduce the map $r: X\times X\rightarrow X\times X,$ such that $r =  \check r\pi,$ 
where $\pi: X\times X\rightarrow X\times X$ 
is the flip map: $\pi(x,y) = (y,x).$ Hence, $r(y,x) =(\sigma_x(y), \tau_y(x)),$ and it satisfies 
the set-theoretic Yang-Baxter equation:
\begin{equation}
r_{12}\ r_{13}\  r_{23} = r_{23}\ r_{13}\ r_{12},  \label{YBE}
\end{equation} 
where we denote $r_{12}(y,x,w) = (\sigma_x(y), \tau_y(x),w),$ $r_{23}(w,y,x) = (w, \sigma_x(y), \tau_y(x))$ 
and \\  $r_{13}(y,w,x) = (\sigma_x(y), w,\tau_y(x)).$ If $\check r$ is involutive then $r$ satisfies $r_{12} r_{21} = \mbox{id}_{X\times X}$ and is called {\it reversible}.

After having introduced the set-theoretic braid or Yang-Baxter equation, which is the main mathematical object in the present work, we may state the
main aim of this paper, which is the review of basic algebraic structures associated to solutions of the set-theoretic Yang-Baxter equation. Specifically, the key objectives of each of the subsequent sections of this article are presented below.

\begin{itemize}

\item In Section 2 we present basic definitions and examples of fundamental algebraic structures 
associated to solutions of the set-theoretic Yang-Baxter equation. These structures are self distributive 
and are known as shelves, racks and quandles \cite{Jo82, Matv}. They also satisfy axioms
analogous to the Reidemeister moves used to manipulate knot diagrams and are associated to link
invariants. Other fundamental structures introduced more recently are the so-called (skew) braces 
invented precisely for the study of generic solutions of the 
set-theoretic Yang-Baxter equation  \cite{Ru05, Ru07, Ru19, GuaVen}.

\item In Section 3 we focus on Baxterized involutive set-theoretic solutions of the braid equation coming from involutive set-theoretic solutions. 
We present some basic properties of these solutions and we then construct the associated quantum algebra  
via the FRT (Faddeev-Reshetikhin-Takhtajan) construction \cite{FRT}. The Yangian is a special case within this class of quantum algebras.
We also present information on the construction of a new class of integrable
quantum spin chain-like systems associated to set-theoretic solutions (see also \cite{DoiSmo2, DoiSmo1} for a more detailed exposition). 

\item In Section 4 we present the new findings on Drinfel'd twists
for involutive, non-degenerate, set-theoretic solutions. 
We focus on a simple, but characteristic example of set-theoretic solution of the YBE known as Lyubashenko’s solution.
We show that Lyubashenko's solutions can be obtained from the permutation operator via a simple Drinfel'd twist. We derive the simple twist as well as explicit
expressions for the $n$-fold twist. We also present the action of the twist on the Yangian's coproducts 
and the derivation of the twisted co-products associated to the Baxterized Lyubashenko solution.

\item In Section 5,
 which is divided in three subsections, we focus on the Hopf algebras associated to set-theoretic solutions \cite{Doikoutw, DoGhVl, DoRy22, DoRySt}. Specifically, in the first subsection, we introduce the Yang-Baxter algebras of rack/quandle type  
as quasi-triangular Hopf algebras (see also related results in \cite{EtScSo99, Andru, Lebed, braceh}).
We first introduce and study the rack and quandle algebras ${\cal A},$ we then show that these are Hopf algebras and we systematically construct the associated universal invertible (or non-degenerate) ${\cal R}$-matrix (i.e. ${\cal R}^{-1}$ exists) \cite{DoRySt}.
In the second subsection, we suitably extend the quandle algebra and present the set-theoretic Yang-Baxter algebra, which is also a Hopf algebra. A  suitable universal Drinfel'd twist is introduced in the third subsection and it is shown to be admissible.
Then the universal set-theoretic ${\cal R}$-matrices are derived as twists, and we conclude that invertible, universal set-theoretic ${\cal R}$-matrices are coming from universal rack ${\cal R}$-matrices via the admissible Drinfel'd twist. The fundamental representation of the universal ${\cal R}$-matrices are the linearized versions of rack and general set-theoretic solutions, and it is consequently shown that all involutive set-theoretic solutions of the braid equation are coming from the permutation operator via the set-theoretic Drinfel'd twist.
A detailed analysis on these results is presented in \cite{Doikoutw, DoRySt}.

\item In Section 6 using the notion of the admissible Drinfel'd twist as well as a certain algebraic structure associated to any generic  set-theoretic solution called the {\it structure group} we are able to extract explicit invertible set-theoretic solutions from quandles. We provide several explicit examples of solutions emerging from distinct quandles.
\end{itemize}

\section{Preliminaries} 

\noindent Before we present the fundamental algebraic structures associated to set-theoretic 
solutions we first introduce some simple examples of 
finite set-theoretic solutions and the notion of linearization. 

\begin{exa}
$ $
\begin{enumerate}
\item {\bf Flip map.} Let $X= \{1,2,\ldots, n\},$ the map $\check r: X \times X \to X \times X,$ 
such as $\check r (a,b) = (b,a) $ is set-theoretic solution of the  braid equation. This is the simplest set-theoretic solution.

\item {\bf Lyubashenko's solution.}  
Let $X= \{1,2,\ldots, n\},$ and the map $\check r:  X \times X \to X \times X,$ such as $ \check r (a,b) = (b+c,a-c),$ 
where the addition is defined $mod\ n$ and $c\in \{1,2, \ldots, n-1\}.$ 
Then $\check r$ is a solution of the set-theoretic braid equation.   
\end{enumerate}
Notice that both solutions above are involutive, i.e. $\check r^2 = \mbox{id}.$
\end{exa}

\noindent {\bf Linearization.}
{\it Via the linearization process we will be able to express the maps $\check r:  X \times X \to X \times X$ as $n^2 \times n^2$ matrices $\check r \in \mbox{End}({\mathbb  C}X^{\otimes 2})$ (we slightly abuse the notation and use the same symbol $\check r$ for the matrices). Specifically, consider a free vector space $V= \mathbb{C}X$ of dimension equal to the cardinality of $X$. 
Let  ${\mathbb B} = \{\hat e_a\}_{a\in X}$ be the basis of $V$ and ${\mathbb B}^* = \{\hat e^*_a\}_{a\in X}$ 
be the dual basis: $\hat e_a^* \hat e_b= \delta_{a,b },$ also  $e_{a,b} := \hat e_a  \hat e_b^*.$ 
Then any set-theoretic solution of the braid equation is expressed as an $n^2 \times n^2$ matrix:
\begin{equation}
\check r = \sum_{a,b \in X} e_{a, \sigma_{a}(b)} \otimes e_{b, \tau_{b}(a)}
\end{equation}
Specifically, for a set $X =\{x_1, x_2, \ldots, x_n\},$  the canonical basis of the $n$-dimensional vector space reads as expected:
$\hat e_{x_1} = \begin{pmatrix} 1 \\ 0\\ 0\\ \vdots\\0 \end{pmatrix},$ $\hat e_{x_2} = \begin{pmatrix} 0 \\ 1\\ 0 \\\vdots \\ 0 \end{pmatrix},$ \ldots,
$\hat e_{x_n} = \begin{pmatrix} 0 \\ 0\\ \vdots\\ 0\\ 1 \end{pmatrix}.$ }

\begin{rem} The action of the $\check r$-matrix on the tensor product of any two elements of the basis $\hat e_x, \hat e_y  \in V,$ for all $x,y \in X$ reads as:
\begin{equation}
\check r\ \hat e_{\sigma_x(y)} \otimes \hat e_{\tau_y(x)}   =\hat e_x  \otimes \hat e_y\, 
\quad \& \quad \check r^T\   \hat e_x  \otimes \hat e_y  = \hat e_{\sigma_x(y)} \otimes \hat e_{\tau_y(x)},
\end{equation}
where $^T$ denotes total transposition. Or equivalently for any test function $f(x,y),$ $x,y \in X:$ $(\check r f)(x,y) = f(\sigma_x(y), \tau_y(x)).$
Similarly, the action of  $r = {\cal P} \check r,$ where ${\cal P}$ 
is the permutation operator ${\cal P} = \sum_{x,y \in X} e_{x,y} \otimes e_{y,x},$ 
on $\hat e_x \otimes \hat  e_y\in V \otimes V$ is
\begin{equation}
r\ \hat e_{\sigma_x(y)} \otimes \hat e_{\tau_y(x)}   =\hat e_y  \otimes \hat e_x\, 
\quad \& \quad  r^T\   \hat e_y \otimes \hat e_x  = \hat e_{\sigma_x(y)} \otimes \hat e_{\tau_y(x)} .
\end{equation}
\end{rem}
We consider for the purposes of this article only finite sets. 
The linerization process can be formally generalized in the case of 
infinite countable sets as above. In the case of compact sets the use of
functional analysis and the study of kernels of integral operators that represent 
the solution $\check r$ is required.

We consider the simple, but non-trivial example of the Lyubashenko solution \cite{Drin}.
\begin{exa} The Lyubashenko solution may be expressed as a $n^2 \times n^2$ matrix: 
\begin{equation}
\check r^{(c)} = \sum_{a,b =1}^n e_{a, b+c} \otimes e_{b, a-c}.
\end{equation}
The addition is defined $mod\ n,$ and for all $x, y \in \{1,2, \ldots, n\}$ and for a fixed $c \in \{1,2,\ldots, n-1\},$
\begin{equation}
    \check r^{(c)}\  \hat e_{y+c}  \otimes \hat e_{x-c}  = \hat e_x  \otimes \hat e_y  , \quad 
    r^{(c)}\ \hat e_{y+c}  \otimes \hat e_{x-c}  =\hat  e_y \otimes \hat e_x.    
    \end{equation}

We focus here on the case $n=3$ and present the two distinct solutions below as $9\times 9$ matrices:
\begin{enumerate}
    \item $c=1,$ $\check r^{(1)} = \sum_{a,b} e_{a, b+1} \otimes e_{b, a-1}.$  
 We write all the nine non-zero terms:
\begin{eqnarray}
\check r^{(1)}  &=&
e_{1, 2} \otimes e_{1, 3} + e_{2, 3} \otimes e_{2,1} + e_{3, 1} \otimes e_{3,2}
+ e_{1, 3} \otimes e_{2,3}+ e_{1,1} \otimes e_{3,3} \nonumber\\ &+& e_{2,1 } \otimes e_{3,1} 
+ e_{2,2 } \otimes e_{1,1}+ e_{3,2 } \otimes e_{1,2}+ e_{3,3 } \otimes e_{2,2}. \nonumber
\end{eqnarray}

\item $c=2,$ $\check r^{(2)} = \sum_{a,b} e_{a, b+2} \otimes e_{b, a-2},$ and explicitly,
\begin{eqnarray}
\check r^{(2)}  &=&
e_{1, 3} \otimes e_{1, 2} + e_{2, 1} \otimes e_{2,3} + e_{3, 2} \otimes e_{3,1}
+ e_{1, 1} \otimes e_{2,2} \nonumber\\ &+& e_{1,2} \otimes e_{3,1}  + e_{2,2} \otimes e_{3,3} 
+ e_{2,3} \otimes e_{1,3}+ e_{3,3} \otimes e_{1,1}+ e_{3,1} \otimes e_{2,1}. \nonumber
\end{eqnarray}

\end{enumerate}
The explicit $9\times 9$ matrices are,

\begin{equation} 
\check r^{(1)}= 
\begin{pmatrix} 0 & 0  &  0 &  0& 0& 1& 0& 0&0 \\
 0& 0& 0& 0 &0 & 0& 0&0 &1  \\
  0& 0& 1& 0& 0& 0& 0 &0 &0 \\
   0& 0& 0& 1 &0 & 0& 0&0&0 \\
 0& 0& 0& 0 &0 & 0& 1&0&0 \\    
 1& 0& 0& 0 &0 & 0& 0&0 &0\\  
  0& 0& 0& 0 &1 & 0& 0&0 &0\\
   0& 0& 0& 0 &0 & 0& 0&1 &0\\
    0& 1& 0& 0 &0 & 0& 0&0&0\\
    \end{pmatrix}, \quad \check r^{(2)}= 
\begin{pmatrix} 0 & 0  &  0 &  0& 0& 0& 0& 1&0 \\
 0& 1& 0& 0 &0 & 0& 0&0 &0  \\
  0& 0& 0& 0& 1& 0& 0 &0 &0 \\
   0& 0& 0& 0 &0 & 0& 0&0&1 \\
 0& 0& 1& 0 &0 & 0& 0&0&0 \\    
 0& 0& 0& 0 &0 & 1& 0&0 &0\\  
  0& 0& 0& 0 &0 & 0& 1&0 &0\\
   1& 0& 0& 0 &0 & 0& 0&0 &0\\
    0& 0& 0& 1 &0 & 0& 0&0&0\\
    \end{pmatrix}. \nonumber
\end{equation}
Notice that $\check r^{(2)} = {\cal P} {\check r}^{(1)} {\cal P},$ where ${\cal P}$ is the permutation operator. In general, for any $n \in \{1,2, \ldots\}$ we observe that $\check r^{(n-k)} = {\cal P} {\check r}^{(k)} {\cal P},$ $k\in \{1,2,\ldots, [{n\over 2}]\}.$

\end{exa}

\subsection{Self distributivity $\&$ special non-involutive solutions}

\noindent  To describe non-involutive solutions of the braid equation we introduce certain algebraic structures that satisfy a self distributivity condition. Self distributive structures, such as shelves, racks $\&$ quandles \cite{Shelf-history2, Jo82, Matv} 
satisfy axioms analogous 
to the Reidemeister moves used to manipulate knot diagrams and are associated to link invariants (see also biracks, biquandles), 
and the coloring of links, i.e. a knot is tri-colored or not. According to Alexander's theorem all links are closed braids, hence these self distributive structures lead naturally to special non-involutive, set-theoretic solutions of the braid equation. Other algebraic structures naturally connected to bijective, non-degenerate solutions are the \emph{bi-racks}. These are algebraic structures that appear in low-dimensional topology that are associated to link diagrams, and are invariant under the generalized Reidemeister moves for virtual knots and links, see \cite{FeJSKa04}. We provide in what follows some preliminaries on left shelves, racks and quandles. For the first systematic study of shelves, we refer the interested reader to \cite{Shelf-history2}. For recent reviews on self-distributive structures the interested reader is referred to \cite{Rev1, Rev2, Rev3} and \cite{Phys} for potential physical applications.

\begin{defn}
    Let $X$ be a non-empty set and\, $\triangleright$\, a binary operation on $X$. Then, the pair $\left(X,\,\triangleright\right)$ 
is said to be a \emph{left shelf} if\, $\triangleright$\, is left self-distributive, namely, the identity
    \begin{equation}\label{eq:shelf}
        a\triangleright\left(b\triangleright c\right)
        = \left(a\triangleright b\right)\triangleright\left(a\triangleright c\right) 
    \end{equation}
    is satisfied, for all $a,b,c\in X$. Moreover, a left shelf $\left(X,\,\triangleright\right)$ is called 
    \begin{enumerate}
        \item  a \emph{left spindle} if $a\triangleright a = a$, for all $a\in X$;
        \item  a \emph{left rack} if $\left(X,\,\triangleright\right)$ is a \emph{left quasigroup}, i.e., the maps $L_a:X\to X$ defined by $L_a\left(b\right):= a\triangleright b$, 
        for all $b\in X$, are bijective, for every $a\in X$.
        \item a \emph{quandle} if $\left(X,\,\triangleright\right)$ is both a left spindle and a left rack.
    \end{enumerate}
\end{defn}

We are mostly interested in racks and quandles here, given that we always require invertible solutions of 
the Yang-Baxter equation. We provide below some fundamental known cases of quandles and racks (see also \cite{Rev1, Rev2, Rev3}):
\begin{enumerate}[{(a)}]
   \item {{\bf Conjugate quandle.} }
   Let {$(X, \cdot)$} be a group and define $\triangleright: X \times X \to X,$ 
such that {$a\triangleright b =a^{-1} \cdot b \cdot a.$} Then $(X, \triangleright)$ is a quandle.    
    
    \item {\bf Core quandle.} Let {$(X, \cdot)$} be a group and $\triangleright: X \times X \to X,$ 
such that {$a\triangleright b =a \cdot b^{-1} \cdot a.$} Then $(X, \triangleright)$ is a quandle.

\item {\bf Alexander (affine) quandle.} Let $Q$ be a ${\mathbb Z}[t, t^{-1}]$ ring module and 
$\triangleright: Q \times Q \to Q,$ $a\triangleright b = (1-t)a +bt,$ then $(Q, \triangleright)$ is a quandle. \\
Or, let $X$ be a non empty set equipped with two group operations, $+$ and $\circ,$ 
such that $a\circ (b+c) = a\circ b -a + a\circ c.$ This is a so-called left skew brace. 
The precise definition is given later in the text. 
Then define $\triangleright: X \times X \to X,$ such that for a fixed $z\in X$ and for all $a,b \in X,$ 
$a\triangleright b = z- a \circ z +b \circ z -z +a$ and 
$(a+b) \circ z =a \circ z - z + a \circ b.$



    \item {\bf Rack, but not quandle.} Let $(G,\cdot)$ be a group and define 
    $\triangleright: G \times G \to G,$ such that $a\triangleright b = b\cdot a^{-1} \cdot x \cdot a,$ 
    where $x\in G$ is fixed. Then $(G, \triangleright)$ is a rack, but not a quandle.   \end{enumerate}

We also present below some concrete examples of finite quandles:
\begin{exa}  \label{exx}

$ $
    \begin{enumerate}
    \item {\bf The dihedral quandle.}
     Let $i,j \in X = {\mathbb Z}_n$ and define $\triangleright: X \times X \to X,$ such that $i\triangleright j = 2i -j$ $mod\ n:$  $(X, \triangleright)$ is a quandle. 
     This is a core quandle with an abelian group. An explicit table of the action $\triangleright$ is presented below for $n=3$
and $X = \{x_1=0, x_2=1, x_3=2 \}:$
\begin{center}
\begin{tabular}{ |c|c|c|c| } 
\hline
$\triangleright$ & $x_1$ &  $x_2$  & $x_3$\\
\hline
\multirow{3}{2em}{\ $x_1$ \\ \ $x_2$\\ \ $x_3$ } 
& $x_1$ & $x_3$ & $x_2$ \\ 
& $x_3$ & $x_2$ & $x_1$  \\ 
& $x_2$ & $x_1$ & $x_3$ \\
\hline
\end{tabular}
\end{center}
\begin{center}
{Table 1}
\end{center}

Although the table above comes from the dihedral quandle, another 
simple representation exists, $\rho: X \to \mbox{End}(\mathbb{C}X),$ such that

\begin{equation}
{x_1 \mapsto \begin{pmatrix}  1& 0&0\\
0&0&1\\
0&1&0
\end{pmatrix},  \quad  x_2 \mapsto 
\begin{pmatrix}  0& 0&1\\
0&1&0\\
1&0&0
\end{pmatrix},} \quad x_3 \mapsto \begin{pmatrix}  0& 1&0\\
1&0&0\\
0&0&1
\end{pmatrix}.  \nonumber \end{equation}
$\rho(x\triangleright y)= \rho(x)^{-1} \bullet \rho(y) \bullet \rho(x)$ for all $x,y \in X$ 
and $\bullet$ is the usual matrix multiplication. 

$ $
    \item {\bf The tetrahedron quandle.} Let $X = \{1,2,3,4\}$ and define $\triangleright: X \times X \to X,$ 
    such that $1\triangleright = (234),$ $2\triangleright = (143),$ $3\triangleright = (124)$ and $4\triangleright =(132).$ 
    This is also a cyclic quandle.  We construct below the explicit table of the action  $\triangleright:$
\begin{center}
\begin{tabular}{ |c|c|c|c|c| } 
\hline
$\triangleright$ & $x_1$ &  $x_2$  & $x_3$ & $x_4$ \\
\hline
\multirow{3}{2em}{\ $x_1$ \\ \ $x_2$\\ \ $x_3$ \\ \ $x_4$ } 
& $x_1$ & $x_3$ & $x_4$ & $x_2$\\ 
& $x_4$ & $x_2$ & $x_1$ & $x_3$\\ 
& $x_2$ & $x_4$ & $x_3$  & $x_1$ \\
& $x_3$ & $x_1$ & $x_2$  & $x_4$ \\
\hline
\end{tabular}
\end{center}
\begin{center}
{Table 2}
\end{center}

    $ $
    
    \item {\bf Example of affine quandle.}  
    Let $X= U_m $ denote the set of odd integers mod $2^m,$ 
$m \in {\mathbb N},$ and the two operation be $+_1,$ such that $a+_1b = a -1+b$ and $\circ,$ where $+$ and $\circ$ are the usual addition and  multiplication. Then a quandle is obtained as in case (c) (Alexander's quandle), i.e. for a fixed $z\in {U_m}$ define, $a\triangleright b = z -_1 a\circ z +_1 b \circ  z {-_1} z +_1 a.$ Specifically, the sets are: 1. for $m=1,$ $U_1 =\{1\},$ 2. for $m=2,$ $U_2 = \{1,\ 3\},$ 3. for $m=3,$ $U_3
= \{1,\ 3,\ 5,\ 7\},$ etc.

Recalling that $a+_1 b = a- 1+b$ and that $(X,+)$ in an abelian group we conclude that $a\triangleright b = -a\circ z + b\circ z + a.$  For instance for
${m=3}$ {($X= \{x_1 =1,\ x_2=3,\ x_3=5,\ x_4=7\}$)} and by choosing for example  $z=3$ we obtain the following table

\begin{center}
\begin{tabular}{ |c|c|c|c|c| } 
\hline
$\triangleright$ & $x_1$ &  $x_2$  & $x_3$ & $x_4$ \\
\hline
\multirow{3}{2em}{\ $x_1$ \\ \ $x_2$\\ \ $x_3$ \\ \ $x_4$ } 
& $x_1$ & $x_4$ & $x_3$ & $x_2$\\ 
& $x_3$ & $x_2$ & $x_1$ & $x_4$\\ 
& $x_1$ & $x_4$ & $x_3$  & $x_2$ \\
& $x_3$ & $x_2$ & $x_1$  & $x_4$ \\
\hline
\end{tabular}
\end{center}

\begin{center}
{Table 3}
\end{center}
\end{enumerate}
\end{exa}
Notice here the distributivity rule between $+_1$ and $\circ$ in case (3) of Example \ref{exx} of an affine quandle. 
Indeed, we observe that for all $a,b,c \in X$ $a \circ (b  +_1 c ) = a \circ b -_1 a +_1 a\circ c$ (see also later in the text the 
definition of skew braces, Definition \ref{defbrace}).

We recall now a fundamental statement regarding shelves and solutions of the set-theoretic Yang-Baxter equation.
\begin{pro} \label{shelf2}
We define 
the binary operation $\triangleright:  X \times X \to X,$ $(a,b) \mapsto a \triangleright b.$ Then $\check r: X \times X \to X \times X$, such that for all $a,b \in X, $  $\check r(a,b) = (b, b \triangleright a)$ is a solution of the set-theoretic braid equation if and only if $(X, \triangleright)$ is a shelf.
\end{pro}
\begin{proof} The proof is straightforward by direct substitution in the Yang-Baxter equation and comparison between LHS and RHS (a graphical depiction of the proof is given below in Figure 4).
\end{proof}

\begin{rem} \label{rack2} 
If $\check r: X \times X \to X \times X$, such that for all $a,b \in X, $  $\check r(a,b) = (b, b \triangleright a)$ 
is an invertible braid solution then $(X, \triangleright)$ is a rack (or a quandle).
\end{rem}

The graphical representation of the shelve solution $\check r(a,b) = (b, b\triangleright a)$:

\begin{figure}[h]
\centering
\begin{tikzpicture}[scale=1.3]
\draw [rounded corners](0,0)--(0,0.25)--(1,0.75)--(1,1);
\draw [rounded corners](1,0)--(1,0.25)--(0.55,0.4);
\draw [rounded corners](0,1)--(0,0.75)--(0.35,0.55);
\node  at (0,-0.2)  {$b$};
\node  at (1,-0.2)  {$b \triangleright a$};
\node  at (-.1,1) [above] {$a$};
\node  at (1.1,1) [above] {$b$};
\end{tikzpicture}
\end{figure}

We are interested here on invertible solutions of the braid equation so we are focusing on rack solutions.
We note that the inverse of $\check r$ above is $\check r^{-1}: X \times X \to X \times X,$ 
$\check r^{-1} (a,b) = (a \triangleright^{-1} b, a),$  such that $a\triangleright(a\triangleright^{-1}b) = a \triangleright^{-1}(a\triangleright b)=b$ for all $a,b \in X.$ Notice also that a different map denoted as $\check r': X\times X \to X\times X,$ such that
$\check r'(a,b) = (a \triangleright b, a)$ is also a solution of the braid equation. $\check r'$ is graphically 
depicted below in Figure 3 together with the braid relation in Figure 4.
\begin{figure}[h]
\centering
\begin{tikzpicture}[scale=1.3]
\draw [rounded corners](0,0)--(0,0.25)--(0.4,0.4);
\draw [rounded corners] (0.6,0.6)--(1,0.75)--(1,1);
\draw [rounded corners](1,0)--(1,0.25)--(0,0.75)--(0,1);
\node  at (0,-0.2)   {${a \triangleright b}$};
\node  at (1,-0.2)  {$a$};
\node  at (-.1,1) [above] {$a$};
\node  at (1.1,1) [above] {$b$};
\node  at (2,.5){ };
\end{tikzpicture}
\begin{center}
{3. The shelve braiding}
\end{center}
$ $

$ $

\begin{tikzpicture}[xscale=1,yscale=1]
\draw [rounded corners](0,0)--(0,0.25)--(0.4,0.4);
\draw [rounded corners](0.6,0.6)--(1,0.75)--(1,1.25)--(1.4,1.4);
\draw [rounded corners](1.6,1.6)--(2,1.75)--(2,3);
\draw [rounded corners](1,0)--(1,0.25)--(0,0.75)--(0,2.25)--(0.4,2.4);
\draw [rounded corners](0.6,2.6)--(1,2.75)--(1,3);
\draw [rounded corners](2,0)--(2,1.25)--(1,1.75)--(1,2.25)--(0,2.75)--(0,3);
\node  at (-0.7,-0.4)  {$\blue{(a\triangleright b)\triangleright (a \triangleright c)}$};
\node  at (1,-0.4)  {${a\triangleright b}$};
\node  at (2.0,-0.4)  {$a$};
\node  at (-.2,3.0) [above] {$a$};
\node  at (1.1,3.0) [above] {$b$};
\node  at (2.1,3.0) [above] {$c$};
\node  at (3.8,1.5){\Large $\overset{}{=}$};
\node  at (4,-.1){ };
\label{P:YBE1}\end{tikzpicture}
\begin{tikzpicture}[xscale=1,yscale=1]
\draw [rounded corners](1,1)--(1,1.25)--(1.4,1.4);
\draw [rounded corners](1.6,1.6)--(2,1.75)--(2,3.25)--(1,3.75)--(1,4);
\draw [rounded corners](0,1)--(0,2.25)--(0.4,2.4);
\draw [rounded corners](0.6,2.6)--(1,2.75)--(1,3.25)--(1.4,3.4);
\draw [rounded corners](1.6,3.6)--(2,3.75)--(2,4);
\draw [rounded corners](2,1)--(2,1.25)--(1,1.75)--(1,2.25)--(0,2.75)--(0,4);
\node  at (-0.4,0.7)   {$\blue{a \triangleright (b \triangleright c)}$};
\node  at (1,0.7)  {${a\triangleright b}$};
\node  at (2.0,0.7)  {$a$};
\node  at (-.2,4.0) [above] {$a$};
\node  at (1.1,4.0) [above] {$b$};
\node  at (2.1,4.0) [above] {$c$};
\node  at (2.3,0.9){ };
\end{tikzpicture}
\begin{center}{4. Shelve solution of the braid equation}
\end{center}
\end{figure}

\begin{exa} We express the solution of the braid equation associated to the dihedral and tetrahedron quandle of Example \ref{exx} as $9 \times 9$ and $16\times 16$ matrices respectively. Recall from linearization we obtain, $\check r$ as a matrix, $\check r = \sum_{x,y\in X}e_{x, y} \otimes e_{y, y\triangleright x },$ where $e_{x,y}$ are the elementary $n\times n$ matrices $e_{x,y} = e_x e_y^T$ ($^T$ denotes transposition). In our examples $n=3, 4$.
\begin{enumerate}
\item (A non-involutive solution from the dihedral quandle).  
More specifically, from Table 1
 \begin{eqnarray}
 \check r &=& \sum_{j=1}^3e_{x_j, x_j} \otimes e_{x_j, x_j}+ e_{x_1, x_2} \otimes  e_{x_2, x_3}+ e_{x_2, x_1} \otimes  e_{x_1, x_3}+e_{x_2, x_3} \otimes  e_{x_3, x_1} 
 \nonumber\\ & + & e_{x_3, x_2} \otimes  e_{x_2, x_1}  + e_{x_1, x_3} \otimes  e_{x_3, x_2} + e_{x_3, x_1} \otimes  e_{x_1, x_2}. \nonumber
 \end{eqnarray}
Then $\check r$ is explicitly expressed as a $9\times 9$ matrix,
\begin{equation} 
\check r= 
\begin{pmatrix} 1 & 0  &  0 &  0& 0& 0& 0& 0&0 \\
 0& 0& 0& 0 &0 & 1& 0&0 &0  \\
  0& 0& 0& 0& 0& 0& 0 &1 &0 \\
   0& 0& 1& 0 &0 & 0& 0&0&0 \\
 0& 0& 0& 0 &1 & 0& 0&0&0 \\    
 0& 0& 0& 0 &0 & 0& 1&0 &0\\  
  0& 1& 0& 0 &0 & 0& 0&0 &0\\
   0& 0& 0& 1 &0 & 0& 0&0 &0\\
    0& 0& 0& 0 &0 & 0& 0&0&1\\
    \end{pmatrix}. \nonumber
\end{equation}

\item (A non-involutive solution from the tetrahedron quandle).  
More specifically, from Table 2
 \begin{eqnarray}
 \check r &=& \sum_{j=1}^4 e_{x_j, x_j} \otimes e_{x_j, x_j}+ e_{x_1, x_2} \otimes  e_{x_2, x_4}+ e_{x_2, x_1} \otimes  e_{x_1, x_3}+e_{x_2, x_3} \otimes  e_{x_3, x_4} 
 \nonumber\\ & + & e_{x_3, x_2} \otimes  e_{x_2, x_1}  + e_{x_1, x_3} \otimes  e_{x_3, x_2} + e_{x_3, x_1} \otimes  e_{x_1, x_4} 
  + e_{x_2, x_4} \otimes  e_{x_4, x_1 }
  \nonumber\\ & + & e_{x_4, x_2} \otimes  e_{x_2, x_3 }  + e_{x_3, x_4} \otimes  e_{x_4, x_2 } + e_{x_4, x_3} \otimes  e_{x_3, x_1 } + e_{x_1, x_4} \otimes  e_{x_4, x_3 }  + e_{x_4, x_1} \otimes  e_{x_1, x_2 }.  \nonumber
 \end{eqnarray}
 \end{enumerate}
\end{exa}

\subsection{Skew braces $\&$ generic set-theoretic solutions}
It is useful to recall at this point  the definition of (skew) braces \cite{Ru05}-\cite{Ru19}, \cite{CeJeOk14, GuaVen} as 
this will allow us to derive generic solutions of the set-theoretic 
Yang-Baxter equation of the type $r: X \times X \to X\times X,$ $r(b,a) = (\sigma_a(b), \tau_b(a))$.
\begin{defn} \label{defbrace} 
A {\it left skew brace} is a set $B$ together with two group operations $+,\circ :B\times B\to B$, 
the first is called addition and the second is called multiplication, such that for all $ a,b,c\in B$,
\begin{equation}\label{def:dis}
a\circ (b+c)=a\circ b-a+a\circ c.
\end{equation}
If $+$ is an abelian group operation, then $B$ is called a 
{\it left brace}.
Moreover, if $B$ is a left skew brace and for all $ a,b,c\in B$ $(b+c)\circ a=b\circ a-a+c \circ a$, then $B$ is called a 
{\it two sided skew brace.} Analogously if $+$ is abelian and $B$ is a skew brace, then $B$ is called a {\it two sided brace}.
\end{defn}
The additive identity of a skew brace $B$ will be denoted by $0$ and the multiplicative identity by $1$.  
In every skew brace $0=1$.

From now on when we say skew brace we mean left skew brace. Some useful examples of braces are presented below:

\begin{exa}[See \cite{BrzRyb:con} Corollary 3.14]\label{ex:cyclicbraces}
Let $\mathrm{U}(\mathbb{Z}/2^m\mathbb{Z})$ denote a set of invertible integers modulo $2^m$, for some $m \in \mathbb{N}$. 
Then a triple $(\mathrm{U}(\mathbb{Z}/2^m\mathbb{Z}),+_1,\circ )$ is a brace, where $a+_1b=a-1+b,$ 
for all $a,b\in \mathrm{U}(\mathbb{Z}/2^m\mathbb{Z})$, $+$ and $\circ $ are addition and multiplication of integer numbers modulo 
$2^m$, respectively. 
\end{exa}

\begin{exa}[See \cite{BrzRybMer} Example 5.7]
Let us consider a ring $\mathbb{Z}/8\mathbb{Z}$. A triple $$\left(\mathrm{OM}:=\left \{\begin{pmatrix}
a & b\\
c & d\end{pmatrix}\ |\ a,d\in \{1,3,5,7\},\  b,c\in\{0,2,4,6\}
\right \},+_\mathbb{I},\circ\right)$$ is a brace, where $(A,B)\overset{+_{\mathbb{I}}}{\longmapsto} A-\mathbb{I}+B$, $(A,B)\overset{\circ}{\longmapsto} A\cdot B$, and $+,\cdot$ are addition and multiplication of two by two matrices over $\mathbb{Z}/8\mathbb{Z}$, respectively. 
\end{exa}

\begin{exa}[See \cite{BrzRybMer} Example 5.6 or \cite{BrzRyb:con} Example 3.15]\label{ex:fractions} 
Let us consider a set $\mathrm{Odd}:=\big\{\frac{2n+1}{2k+1}\ |\ n,k\in\mathbb{Z}\big\}$ together with two binary operations 
$(a,b)\overset{+_1}{\longmapsto}a-1+b$ and $(a,b)\overset{\circ}{\longmapsto}a\cdot b$, where $+,\cdot$ 
are addition and multiplication 
of rational numbers, respectively. The triple $(\mathrm{Odd},+_1,\circ)$ is a brace. 
\end{exa}

We recall now the basic conditions associated to any generic 
solution of the set-theoretic Yang-Baxter equation as they will be used in our analysis here.
\begin{pro} 
Let $X$ be a non-empty set, and define for all $a,b\in X,$ the maps $\sigma_a, \tau_b: X \to X,$ 
$b \mapsto \sigma_a(b)$ and $a \mapsto \tau_b(a).$ 
Then $r:X \times X \to X \times X$, 
such that for all $a,b\in X,$ $r(b,a) = (\sigma_a(b), \tau_b(a))$ is a solution of the set-theoretic Yang-Baxter equation 
if and only if for all $a,b,c \in X,$
\begin{eqnarray}
&&  \sigma_a(\sigma_b(c)) = \sigma_{\sigma_a\left(b\right)}(\sigma_{\tau_b\left(a\right)}(c)) \label{C1}\\ 
&& \tau_c(\tau_b(a)) =\tau_{\tau_c\left(b\right)}(\tau_{\sigma_b\left(c\right)}(a)) \label{C2}\\
&& \sigma_{\tau_{\sigma_b\left(c\right)}\left(a\right)}
    (\tau_c\left(b\right)) 
   = \tau_{\sigma_{\tau_b\left(a\right)}\left(c\right)}(\sigma_a\left(b\right) ). \label{C3} \end{eqnarray}
\end{pro}
\begin{proof} Let $r$ be a solution of the Yang-Baxter equation.
We compute explicitly the LHS and RHS of the parametric Yang-Baxter equation. 
The LHS of the Yang-Baxter equation gives, $a,b,c \in X,$ 
    \begin{equation}
    r_{12} \ r_{13}\ r_{23}(c,b,a)  = \big (\sigma_{\sigma_a\left(b\right)}(\sigma_{\tau_b\left(a\right)}(c)),\ \tau_{\sigma_{\tau_b\left(a\right)}\left(c\right)}(\sigma_a\left(b\right) ),\  \tau_c(\tau_b(a)) \big ), \label{lhs2} \end{equation}
    whereas the RHS gives
    \begin{equation}    
    r_{23}\ r_{13}\  r_{12}(c,b,a)  =  \big (\sigma_a(\sigma_b(c)),\ \sigma_{\tau_{\sigma_b\left(c\right)}\left(a\right)}
    (\tau_c\left(b\right)),\  \tau_{\tau_c\left(b\right)}(\tau_{\sigma_b\left(c\right)}(a))\big ). \label{rhs2} \end{equation}
 By equating (\ref{lhs2}) and (\ref{rhs2}) we arrive at (\ref{C1})-(\ref{C3}). 
 And conversely, if conditions (\ref{C1})-(\ref{C3}) are satisfied then $r$ automatically satisfies the Yang-Baxter equation.
\end{proof}

We may now prove a key proposition on generic set-theoretic solutions coming from skew {braces} \cite{Ru05, Ru07,CeJeOk14, GuaVen}.
\begin{pro}  \label{pp0} Let $(X,+,\circ)$ be a skew brace and let $\sigma_a(b):=  - a +a \circ b$ and $a\circ b = \sigma_a(b) \circ \tau_b(a)$ for all $a,b \in X.$ Then the map $r: X \times X \to X \times X,$ $r(b,a) = (\sigma_a(b), \tau_{b}(a))$ is a solution of the set theoretic Yang-Baxter equation.
\end{pro}
\begin{proof}
    In order to prove that $r$ is a solution of the Yang-Baxter equation it suffices to show that the three conditions (\ref{C1})-(\ref{C3}) hold.

Before we proceed with the proof we observe that the distributivity condition in skew braces for all $a,b,c\in X,$ $a\circ( b+c) = a\circ b -a + a\circ c$ is equivalent to $a\circ (b -c +d) = a \circ b -a \circ c+a \circ d $ (see also \cite{DoRy22}) for a detailed proof).

We first show condition (\ref{C1}), indeed for $a,b, c \in X,$
\begin{eqnarray}
\sigma_a(\sigma_b(c)) &=& -a + a \circ \sigma_b(c) = 
-a +a \circ (-b + b\circ c)= -a \circ b +a \circ b \circ c =\sigma_{a\circ b}(c) \nonumber\\
&=&  \sigma_{\sigma_a(b) \circ \tau_b(a)}(c) = \sigma_{\sigma_a\left(b\right)}(\sigma_{\tau_b\left(a\right)}(c)).
\nonumber
\end{eqnarray}

We show now condition (\ref{C2}),
\begin{equation}
\tau_c(\tau_b(a)) = \sigma_{\tau_b(a)}(c)^{-1} \circ \tau_b(a) \circ c = \sigma_{\tau_b(a)}(c)^{-1} \circ 
\sigma_a(b)^{-1} \circ a \circ b \circ c. \nonumber
\end{equation}
But,
\begin{eqnarray}
\sigma_a(b) \circ \sigma_{\tau_b(a)}(c) &=& \sigma_a(b) \circ ( -\tau_b(a) +\tau_b(a) \circ c) = 
\sigma_a(b) - a\circ b + a\circ b \circ c \nonumber\\
&=&  -a + a\circ b \circ c = \sigma_{a}(b\circ c)
\end{eqnarray}
hence,
\begin{equation}
\tau_c(\tau_b(a)) = \tau_{b\circ c}(a) ={\tau_{\sigma_b(c) \circ \tau_c(b)}} = \tau_{\tau_c\left(b\right)}(\tau_{\sigma_b\left(c\right)}(a)).
\end{equation}
Condition (\ref{C3}) follows from (\ref{C1}), (\ref{C2}) using $a\circ b = \sigma_a(b) \circ \tau_b(a),$ for all $a,b \in X.$ And, this concludes our proof.
\end{proof}

The group $(X, \circ),$ such that for all $a,b \in X,$ $a\circ b = \sigma_a(b) \circ \tau_b(a)$ is called the {\it structure group} of a set-theoretic solution. In the case that $(X, +, \circ)$ is a brace we obtain 
involutive solutions \cite{Ru05, Ru07, Ru19}. According to Rump all involutive solutions are obtained from braces, 
therefore in this article we often call the involutive set-theoretic solutions {\it brace solutions}.

The important fact that will be discussed in Section 5, is that all involutive solutions are obtained from the permutation operator by a suitable Drinfel'd twist (see also \cite{Doikoutw, Sol,LebVen}), whereas all the non-involutive but invertible solutions are obtained from rack solutions via a Drinfel'd twist (degenerate solutions are coming from shelves) \cite{DoRy22, DoRySt}.

\section{Involutive solutions $\&$ Baxterization}

\noindent 
We focus in this section on involutive set-theoretic solutions and derive Baxterized solutions of the 
braid and Yang-Baxter equations. We then identify the quantum algebras associated to these solutions (see also \cite{DoiSmo2, DoiSmo1}).

Recall the Yang-Baxter equation in the braid form  in the presence of spectral parameters 
$\lambda_1,\ \lambda_2$ ($\delta = \lambda_1 - \lambda_2$):
\begin{equation}
\check R_{12}(\delta)\ \check R_{23}(\lambda_1)\ \check R_{12}(\lambda_2) = \check R_{23}(\lambda_2)\
 \check R_{12}(\lambda_1)\ \check R_{23}(\delta), \label{YBE1}
\end{equation}
where $\check R: V \otimes V\to V \otimes V,$  ($V$ is an $n$ dimensional space) and let in general $\check R = \sum_{j} a_j \otimes b_j,$ then in the index notation $\check R =\sum_j a_j \otimes b_j \otimes 1_V,$  $\check R_{23} =\sum_j  1_V \otimes a_j \otimes b_j$ 
and $\check R_{13} =\sum_j a_j \otimes  1_V \otimes b_j.$

We focus here on Baxterized solutions of the form
\begin{equation}
\check R(\lambda) = \lambda \check r + 1_V^{\otimes 2}, \label{braid1}
\end{equation}
where $\check r$ is an involutive solution of the braid equation. Let also, $R = {\mathcal P} \check R$, then
\begin{equation}
R(\lambda)= \lambda r + {\mathcal P}, \label{braid2}
\end{equation}
and $R$ is a solution of the Yang-Baxter equation,
\begin{equation}
 R_{12}(\delta)\  R_{13}(\lambda_1)\  R_{23}(\lambda_2) = R_{23}(\lambda_2)\ R_{13}(\lambda_1)\ R_{12}(\delta). \label{YBE2}
\end{equation}

\begin{rem}
It would be useful for the following Proposition to define the partial transposition. 
Let  $A\in \mbox{End}\big ({\mathbb C}^n \otimes {\mathbb C}^n \big )$
expressed as: $A = \sum_{i, j,k,l=1 }^nA_{ij, kl}\ e_{i,j}\otimes e_{k,l}$.
 We define the {\it partial transposition}  as follows (in the index notation):
\begin{eqnarray}
A_{12}^{t_1} = \sum_{i, j,k,l=1}^n A_{ij, kl}\ e_{i, j}^t \otimes e_{k, l},\
\quad A_{12}^{t_2} = \sum_{i, j,k,l =1}^n A_{ij, kl}\ e_{i, j} \otimes e_{k, l}^t
\end{eqnarray}
where $e_{i,j}^t = e_{j,i}$.
\end{rem}

\begin{pro}\label{555}  The brace $R$-matrix satisfies the following fundamental properties:
\begin{eqnarray}
&&  R_{12}(\lambda)\  R_{21}(-\lambda) = (-\lambda^2 +1)  1_{V}^{\otimes 2}, ~~~~~~~~~~~~~\mbox{{\it Unitarity}} \label{u1}\\
&&  R_{12}^{t_1}(\lambda)\ R_{12}^{t_2}(-\lambda -n) = \lambda(-\lambda -n) 1_{V}^{\otimes 2}, ~~~~~
\mbox{{\it Crossing-unitarity}} \label{u2}\\
&& R_{12}^{t_1 t_2}(\lambda) = R_{21}(\lambda), \label{tt}
\end{eqnarray}
{\it where $^{t_{1,2}}$ denotes transposition on the first, second space respectively.}
\end{pro}
\begin{proof} The proof follows immediately after a straightforward computation \cite{DoiSmo1}.
\end{proof}

\subsection{The quantum algebra associated to braces}
\noindent  We recall the definitions of two quadratic 
algebras ${\mathcal A}$ and ${\mathcal Q}$ associated to solutions of the Yang-Baxter equation, which arise 
from the FRT (Faddeev, Reshetikhin and Takhtajan) construction \cite{FRT}. 
Indeed, from the FRT construction we recall that
given a solution of the braid equation $\check r: V \times V\to V \times V$ (henceforth we consider $V = {\mathbb C}^n$)
the associated quantum algebra ${\mathcal A}$ is a quotient of a 
free associative ${\mathbb C}$-algebra, generated by $\{L_{z,w}|\ x,w \in X\},$ and relations 
\begin{equation}
\check r_{12}\ L_1\ L_2 = L_1\ L_2\ \check r_{12}, \label{RTT0} \
\end{equation} 
where  $\ L  = \sum_{x,y \in X} e_{x,y} \otimes L_{x,y}\in 
\mbox{End}(V) \otimes {\mathcal A}$. 
Recall  the {\it index notation}: $\check r_{12} = \check r \otimes 1_{\mathcal A}$ and
$L_1 = \sum_{z, w \in X} e_{z,w} \otimes 1_V \otimes L_{z,w}, $ $\  L_2= \sum_{z, w \in X} 1_V  \otimes  e_{z,w}  \otimes L_{z,w}.$

From the fundamental relation (\ref{RTT0}) and by considering the set-theoretic solution of the braid equation 
\begin{equation}
\check r = \sum_{x,y \in X} e_{x, \sigma_x(y)} \otimes e_{y,\tau_{y}(x)}, \label{brace}
\end{equation}
we obtain \cite{EtScSo99}:
 \begin{eqnarray}
L_{x, \hat x} L_{y, \hat y} = L_{\sigma_x(y), \sigma_{\hat x}(\hat y)} L_{\tau_y(x), \tau_{\hat y}(\hat x)}.
 \label{q11}
\end{eqnarray} 
Recall that sometimes we call the involutive, set-theoretic solutions, brace solutions, 
because they are all obtained from braces \cite{Ru05, Ru07}.
Given a solution of the braid equation $\check r: V \otimes V \to V \otimes V,$ the quadratic algebra ${\mathcal Q}$ is generated by $\{q_x|\ x \in X\}$  and relations
\begin{equation}
 \check r_{12}\ q_1\ q_2 = q_1\ q_2, \label{RTT2} 
\end{equation} 
where $\ q  = e_x \otimes q_x \in V \otimes {\mathcal Q}.$ Also, $\check r_{12} =\check r \otimes 1_{\mathcal A},$
$\ q_1 = \sum_{z, w \in X} e_{x} \otimes 1_V \otimes q_{x},$ $\ q_2= \sum_{x \in X} 1_V  \otimes  e_{x}  \otimes q_{x}.$
The quadratic relation (\ref{RTT2}) for the set-theoretic solution implies 
 \begin{eqnarray}
q_{x} q_{y} = q_{\sigma_x(y)} q_{\tau_y(x)}, \label{qalg}
\end{eqnarray} 
also obtained in \cite{EtScSo99}.

We now consider the Baxterized solution $\check R(\lambda) = \lambda \check r + 1_{V\otimes V}$, 
where $\check r$ in our analysis here is the set-theoretic solution of the braid equation (\ref{brace}).
Given a parametric solution of the Yang-Baxter equation, the quantum algebra is defined via the fundamental relation \cite{FRT}:
\begin{equation}
\check R_{12}(\lambda_1 -\lambda_2)\ L_1(\lambda_1)\ L_2(\lambda_2) = L_1(\lambda_2)\ L_2(\lambda_1)\ 
\check R_{12}(\lambda_1 -\lambda_2). \label{RTT}
\end{equation}
$\check R(\lambda) \in \mbox{End}({\mathbb C}^{n} \otimes {\mathbb C}^{n})$, $\ L(\lambda) \in 
\mbox{End}({\mathbb C}^{n}) \otimes {\mathfrak A}$, where ${\mathfrak A}$\footnote{Notice that in $L$ 
in addition to the indices 1 and 2 in (\ref{RTT}) there is also an implicit ``quantum index'' $n$ associated to ${\mathfrak A},$ 
which for now is omitted, i.e. one writes $L_{1n},\ L_{2n}$.} is the quantum algebra defined by (\ref{RTT}). 
We shall focus henceforth on Baxterized solutions coming from involutive, set-theoretic solutions. 
The defining relations of the corresponding 
quantum algebra were derived in \cite{DoiSmo2, DoiSmo1}.

The quantum algebra associated to Baxterized solutions coming from braces is defined by generators $L^{(m)}_{z,w},\ z, w \in X$, and defining relations (see also \cite{DoiSmo1})
\begin{eqnarray}
L_{z,w}^{(n)} L_{\hat z, \hat w}^{(m)} - L_{z,w}^{(m)} L_{\hat z, \hat w}^{(n)} &=& 
L^{(m)}_{z, \sigma_w(\hat w)} L^{(n+1)}_{\hat z,\tau_{\hat w}( w)}- L^{(m+1)}_{z, \sigma_w(\hat w)} 
L^{(n)}_{\hat z, \tau_{\hat w}( w)}\nonumber\\ &-& L^{(n+1)}_{ \sigma_z(\hat z),w} 
L^{(m)}_{\tau_{\hat z}( z), \hat w }+ L^{(n)}_{ \sigma_z(\hat z, )w}  L^{(m+1)}_{\tau_{\hat z}( z), \hat w}. \label{fund2}
\end{eqnarray}

The proof is based on the fundamental relation (\ref{RTT}) and the form of the Baxterized brace $R$-matrix (for the detailed proof see \cite{DoiSmo2, DoiSmo1}). Recall also that
in the index notation we define $\check R_{12} = \check R \otimes \mbox{id}_{\mathfrak A}$:
\begin{eqnarray}
&& L_1(\lambda) = \sum_{z, w \in X} e_{z,w} \otimes 1_V \otimes L_{z,w}(\lambda),\ \quad  L_2(\lambda)= \sum_{z, w \in X}
1_V  \otimes  e_{z,w}  \otimes L_{z,w}(\lambda).  \label{def}
\end{eqnarray} 
The exchange relations among the various generators of the affine algebra 
are derived below via (\ref{RTT}). Let us express $L$ as a formal power series expansion 
$L(\lambda) = \sum_{n=0}^{\infty} {L^{(n)} \over \lambda^n}$.
Substituting  expressions (\ref{braid1}), and the $\lambda^{-1}$ expansion in (\ref{RTT}) we obtain
the defining relations of the quantum algebra associated 
to a brace $R$-matrix (we focus on terms $\lambda_1^{-n} \lambda_2^{-m}$):
\begin{eqnarray}
&&  \check r_{12} L_{1}^{(n+1)} L_2^{(m)} -\check  r_{12} L_1^{(n)} L_2^{(m+1)} +  L_1^{(n)} L_2^{(m)} \nonumber\\
&&  = L_1^{(m)} L_{2}^{(n+1)} \check r_{12} -  L_1^{(m+1)} L_2^{(n)}\check r_{12} +  L_1^{(m)} L_2^{(n)}. \label{fund}
\end{eqnarray}
The latter relations immediately lead to the quantum algebra relations (\ref{fund2}), after recalling:
$
L_{1}^{(k)}=\sum_{x,y\in X}e_{x,y}\otimes 1_V \otimes  L_{x,y}^{(k)},$ $ L_{2}^{(k)}=\sum_{x,y\in X} 1_V \otimes 
e_{x,y}\otimes  L^{(k)}_{x,y}, \nonumber
$
and $\check r_{12} = \check r \otimes  \mbox{id}_{\mathfrak A },$ 
$L^{(k)}_{x,y} $ are the generators of the associated quantum algebra.
The quantum algebra is also equipped with a co-product $\Delta: {\mathfrak A} \to {\mathfrak A} \otimes {\mathfrak A}$ \cite{FRT, Drinfeld}. Indeed, we define 
\begin{equation}
(\mbox{id} \otimes \Delta) L(\lambda) := L_{13}(\lambda) L_{12}(\lambda),\  
\end{equation}
which satisfies (\ref{RTT}) and is expressed as $(\mbox{id} \otimes \Delta) L(\lambda) = \sum_{x,y \in X} e_{x,y} \otimes \Delta(L_{x,y}(\lambda)).$

\begin{rem} In the special case $\check r ={\cal P}$ the ${\cal Y}(\mathfrak {gl}_n)$ algebra is recovered:
\begin{equation}
\Big [ L_{i,j}^{(n+1)},\ L_{k,l}^{(m)}\Big ] -\Big [ L_{i,j}^{(n)},\ L_{k,l}^{(m+1)}\Big ] = L_{k,j}^{(m)}L_{i,l}^{(n)}- L_{k,j}^{(n)}L_{i,l}^{(m)}. \label{fund2b}
\end{equation}
\end{rem}

The next natural step is  the classification of solutions of the fundamental relation (\ref{RTT}), 
for the brace quantum algebra.  A first step towards this goal  will be to examine  the fundamental object $L(\lambda)= L_0 + {1\over\lambda} L_1$, and search for finite and infinite representations of the respective elements. 
The classification of $L$-operators  will allow the identification of  new classes of quantum integrable systems, such as the analogues of Toda chains or deformed boson models. A first obvious example to consider is associated to Lyubashenko's solution, which is further discussed later in the manuscript.

\subsection{Integrability: local Hamiltonians}

\noindent Given any solution of the Yang-Baxter equation we define the so-called  monodromy matrix 
$T_{0,12...N}(\lambda) \in \mbox{End}\big ({\mathbb C}^{n} \otimes({\mathbb C}^{n})^{\otimes N}\big )$, 
which is a tensor representation of the quantum group (\ref{RTT}), \cite{FRT} 
\begin{equation}
T_{0,12...N}(\lambda) :=R_{0N}(\lambda) \ldots R_{02}(\lambda)\  R_{01}(\lambda), \label{mono}
\end{equation}
recall $R = {\mathcal P} \check R.$
We define also the  transfer matrix ${\mathfrak t}_{12...N}(\lambda) = tr_0 \big (T_{0,12...N}(\lambda)\big )  \in  
\mbox{End}\big (({\mathbb C}^{n})^{\otimes N}\big )$. 
The monodromy matrix $T$ satisfies (\ref{RTT}), and hence one can show that the transfer matrix provides mutually 
commuting quantities \cite{FRT}: 
(${\mathfrak t}(\lambda) =\lambda^N  \sum_{k} {{\mathfrak t}^{(k)} \over \lambda^k}$)
\begin{equation}
\Big [ {\mathfrak t}(\lambda),\ {\mathfrak t}(\mu)\Big ] =0 \ \Rightarrow\  \Big [ {\mathfrak t}^{(k)},\ {\mathfrak t}^{(l)}\Big ] =0.
\label{invo}
\end{equation}
Note that historically the index $0$ is called ``auxiliary'', whereas the indices $1,2,  \ldots, N$ are called ``quantum'',
and they are usually suppressed for simplicity, i.e. we simply write $T_0(\lambda)$ and ${\mathfrak t}(\lambda)$.  

The following Proposition \ref{Traces} is quite general and holds
for any $R(\lambda) =  \lambda {\mathcal P}\check r +{\mathcal P}$, where $\check r$ 
is an involutive solution of the braid equation and ${\mathcal P}$ is the permutation operator.
The key property that allows the derivation of local Hamiltonians is $R(\lambda = 0) = {\cal P}.$
\begin{pro}\label{Traces}  Consider the $\lambda$-series expansion of the monodromy matrix: $T(\lambda) = 
\lambda^N\sum_{k=0}^N{T^{(k)} \over \lambda^k}$ for any $R(\lambda) = \lambda {\mathcal P}\check r +{\mathcal P}$, where $\check r$  
is an involutive solution of the braid equation.  Let also  $H^{(k)} = {\mathfrak t}^{(k)} ({\mathfrak t}^{(N)})^{-1}$, $\ k = 0, 
\ldots, N-1$ and $H^{(N)} = {\mathfrak t}^{(N)}$, where 
${\mathfrak t}^{(k)} = tr_0 (T_0^{(k)})$. 
Then the commuting quantities, $H^{(k)}$ for $k = 1, \ldots, N-1$, are expressed exclusively in terms 
of the elements $\check r_{n \, n+1}$, $n = 1, \ldots, N-1$, and $\check r_{N \, 1}$.
\end{pro}
\begin{proof} We refer the interested reader to  \cite{DoiSmo1} for the detailed proof.
\end{proof}

The generic first neighbor Hamiltonian is given as
${\cal H} = \sum_{j=1}^N \check r_{j j+1}.$ Higher commuting quantities can be found in \cite{DoiSmo1} 
for periodic spin chains and in \cite{DoiSmo2} for open spin chains. In the special case of set-theoretic 
solutions (\ref{brace}) the local Hamiltonian becomes
\begin{equation}
{\cal H} =   \sum_{j=1}^N \sum_{a,b\in X}  e_{a,\sigma_{a}(b)}^{(j)} e_{b, \tau_{b}(a)}^{(j+1)}. \label{Ham2}
\end{equation}
A simple example within this class is the Lyubashenko solution (see Example 2.1), then the Hamiltonian (\ref{Ham2}) takes the simple form, 
$${\cal H}_c=  \sum_{j=1}^N \sum_{a,b=1}^n  e_{a,b+c}^{(j)} e_{b, a-c}^{(j+1)},$$ where recall $c\in \{1,2, \ldots,n -1\}$ is fixed.

The ultimate goal in the context of quantum 
integrable systems, or any quantum system for that matter, is the identification of the eigenvalues 
and eigenvectors of the corresponding Hamiltonian.
In the frame of quantum integrable systems there exists a set of mutually commuting 
``Hamiltonians'',  guaranteed by the existence of a quantum $R$-matrix that satisfies the Yang-Baxter equation.  An exhaustive analysis of the symmetries of periodic and open quantum spin chains constructed from Baxterized involutive set-theoretic solutions is presented in \cite{DoiSmo2, DoiSmo1}. The hierarchy of periodic and open mutually commuting Hamiltonians is also explicitly derived in \cite{DoiSmo2, DoiSmo1} exclusively in terms of the elements of the symmetric group.

\section{Set-theoretic solutions as Drinfel'd twists: an example}

 \noindent We recall in this section the Drinfel'd twist for set-theoretic solutions.
It was shown in \cite{Doikoutw} that all involutive, set-theoretic solutions 
can be obtained from the permutation operator via suitable twists (see also \cite{Sol} for an analogous non-local map), 
whereas non-involutive, invertible solutions are obtained via
the same twist from rack/quandle solutions \cite{DoRy22, DoRySt}. We consider here a simple example of set-theoretic solution of the braid equation obtained as a simple twist of the permutation to provide a key motivation for the general results 
presented in the subsequent section. 

\subsection*{Simple non-trivial case: Lyubashenko's solution}
\noindent We recall the Lyubashenko solution and show that is immediately 
obtained from the permutation operator as a simple twist.
Although the construction is simple it has significant implications on the associated symmetries of the braid solutions.  
Inspired by the isotropic case a similar construction for the $q$-deformed analogue of Lyubashenko's 
solution is provided in \cite{DoiSmo1, DoiSmo2}.

Before we derive the Lyubashenko solution as a suitable twist we first introduce a useful Lemma.
\begin{lemma} {\label{extra1}} Let $\check r': V \otimes V \to V \otimes V$ ($V$ is an $n$ dimensional vector space) 
satisfy the braid relation and $(\check r')^2 = 1_V^{\otimes 2}$. Let also $u: V \to V$ be an invertible map,
such that $(u \otimes u)  \check r' = \check r'  (u \otimes u)$. 
We define $\check r = (u \otimes 1_V)\check  r' (u^{-1} \otimes 1_V) =
(1_V\otimes u^{-1})\check r' (1_V \otimes u),$ then: 
\begin{enumerate}
\item ${\check r}^2 = 1_V^{\otimes 2}$
\item $\check r$ satisfies the braid relation.
\end{enumerate}
\end{lemma}

\begin{proof} The proof is straightforward \cite{DoiSmo2}.
\end{proof}

\begin{pro}\label{prop1} Let $\tau,\ \sigma: X \to X$, $X= \{1,\ldots, n \}$  be isomorphisms, 
such that $\sigma(\tau(x)) = \tau(\sigma(x)) = x$ and let
$u=\sum_{x \in X} e_{x, \tau(x)}$ and $u^{-1} = \sum_{x \in X}  e_{ \tau(x),x}$.
Then any solution of the type (Lyubashenko's solution)
\begin{equation}
\check r=\sum_{x, y \in X} e_{x, \sigma(y)} \otimes e_{y, \tau(x)}, \label{special1}
\end{equation}
is obtained from the permutation operator ${\cal P}= \sum_{x, y \in X} e_{x,y} \otimes e_{y,x}$ as
\begin{equation}
\check r = (u\otimes 1_V ){\cal P} (u^{-1} \otimes 1_V)= ( 1_V \otimes u^{-1} ) {\cal P} (1_V \otimes u). \label{special1b}
\end{equation}
\end{pro}
\begin{proof}
The proof relies on the definitions 
of ${\cal P},\ u,\ u^{-1},$ the fundamental property $e_{x,y} e_{z,w} = \delta_{y,z} e_{x,w}$
and by straightforward computation.
\end{proof}

Note that $r = {\cal P} \check r$, and consequently $R = {\cal P} \check R$
 take a simple form for this class of solutions:
\begin{equation}
r= u^{-1} \otimes u\ \Rightarrow\  R(\lambda) = \lambda u^{-1} \otimes u + {\cal P}. \label{special2}
\end{equation}





Before we present our findings on the symmetry of Lyubashenko's $\check r$-matrix we first introduce a useful Lemma \cite{DoiSmo1, DoiSmo2}.

\begin{lemma}{\label{extra2}} Let ${\mathfrak l}_{x,y}$ be the generators of the $\mathfrak{gl}_n$ algebra:
\begin{equation}
\Big [{\mathfrak l}_{x,y},  {\mathfrak l}_{z,w}\Big ] = \delta_{y,z}{\mathfrak l}_{x,w} - \delta_{x,w}{\mathfrak l}_{z,y}. \label{gl2}
\end{equation}
The $\mathfrak{gl}_n$ algebra is equipped with a coproduct $\Delta: \mathfrak{gl}_n \to  \mathfrak{gl}_n  \otimes  \mathfrak{gl}_n,$ such that
\begin{equation}
\Delta({\mathfrak l}_{x,y}) = {\mathfrak l}_{x,y} \otimes \mbox{id} +  \mbox{id}  \otimes {\mathfrak l}_{x,y}.
\end{equation}
The $N$-coproduct is obtained by iteration $ \Delta^{(N)} =  ( \Delta^{(N-1)}  \otimes \mbox{id})  \Delta= (\mbox{id} \otimes \Delta^{(N-1)}) \Delta $ and is given as $\Delta^{(N)}({\mathfrak l_{x,y}})=  \sum_{n=1}^N \mbox{id} \otimes \ldots \otimes\underbrace{{\mathfrak l}_{x,y}}_{n^{th}\  \mbox{position}} \otimes \ldots \otimes \mbox{id}$.

Let also ${\cal F}^{(N)}\in \mathfrak{gl}_n^{\otimes N}$ 
be an invertible element (${\cal F}^{(2)}=: {\cal F}$), and define\\ 
$\Delta_T^{(N)}({\mathfrak l}_{x,y}):= {\cal F}^{(N)}\Delta^{(N)}({\mathfrak l}_{x,y}) ({\cal F}^{(N)})^{-1},$ then $\Delta_T^{(N)}({\mathrm l}_{x,y})$ also satisfy the $\mathfrak{gl}_n$ algebraic relations.
\end{lemma}
\begin{proof} The $N$-coproducts satisfy the $\mathfrak{gl}_n$ relations (\ref{gl2}), i.e. 
$\Big [\Delta^{(N)}({\mathfrak l}_{x,y}),  \Delta^{(N)}({\mathfrak l}_{z,w})\Big ] = \delta_{y,z}\Delta^{(N)}({\mathfrak l}_{x,w}) - \delta_{x,w}\Delta^{(N)}({\mathfrak l}_{z,y}).$ By acting from the left with ${\cal F}^{(N)}$ and with $({\cal F}^{N})^{-1}$ from the right in the latter commutator we immediately obtain $\Big [\Delta_T^{(N)}({\mathfrak l}_{x,y}),  \Delta_T^{(N)}({\mathfrak l}_{z,w})\Big ] = \delta_{y,z}\Delta_T^{(N)}({\mathfrak l}_{x,w}) - \delta_{x,w}\Delta_T^{(N)}({\mathfrak l}_{z,y})$.
\end{proof}

\begin{cor}\label{prop2} {\it Let  $\rho: \mathfrak{gl}_n \to \mbox{End}({\mathbb C}^n)$ be the fundamental representation of $\mathfrak{gl}_n,$ such that ${\mathfrak l}_{x,y} \mapsto e_{x,y},$ where recall $e_{x, y}$ are $n \times n$ matrices with elements 
$(e_{x,y})_{z,w}=\delta_{x,z} \delta_{y,w}.$ The special solution $\check r$ (\ref{special1}) is $\mathfrak{gl}_n$ symmetric, i.e.}
\begin{equation}
\Big [ \check r,\ \Delta_i(e_{x,y}) \Big ] =0, ~~~~x,\ y \in X, \label{symm1}
\end{equation}
{\it where we define the ``twisted'' co-products ($i= 1, 2$): }
\begin{eqnarray}
&&\Delta_1(e_{x,y}) = e_{\sigma(x), \sigma(y)} \otimes 1_V +1_V  \otimes e_{x,y}, \nonumber\\
&& \Delta_2(e_{x,y}) =  e_{x,y} \otimes  1_V + 1_V \otimes   e_{\tau(x) ,\tau(y)} ,  \label{symm2}
\end{eqnarray}
($\Delta_1(e_{\tau(x), \tau(y)}) = \Delta_2(e_{x, y})$).
\end{cor}
\begin{proof}
This can be shown using the form of the special class of solutions (\ref{special1}). 
The permutation operator is $\mathfrak{gl}_n$ symmetric, i.e.
\begin{equation}
\Big [ {\cal P},\ \Delta(e_{x,y})\Big ] =0, \label{comm1}
\end{equation}
where the co-products $\Delta(e_{x,y})$  are defined in Lemma (\ref{extra2}) (${\mathfrak l}_{x,y} \mapsto e_{x,y}$).

Let $u= \sum_{x \in X} e_{x, \tau(x)}$, then (\ref{symm1}) immediately follows from (\ref{comm1}) and (\ref{special1b}) after 
acting (\ref{comm1}) from the left and right with $u\otimes 1_V,\  u^{-1}  \otimes 1_V$  or  
$ 1_V \otimes u^{-1},\   1_V \otimes u$ respectively.  $\Delta_{i}(e_{x,y})$ are then defined as
\begin{eqnarray} 
&& \Delta_1(e_{x,y}) = u e_{x,y} u^{-1} \otimes 1_V + 1_V \otimes e_{x,y}, \nonumber\\ 
&& \Delta_2(e_{x,y}) = e_{x,y} \otimes 1_V  + 1_V \otimes u^{-1}   e_{x,y}u \label{sim1}
\end{eqnarray}
and explicitly given by (\ref{symm2}). Indeed, $u e_{x,y} u^{-1} = e_{\sigma(x), \sigma(y)}$ and  $u^{-1} e_{x,y} u = e_{\tau(x), \tau(y)}.$

According to Lemma \ref{extra2} $\Delta_i(e_{x, y})$ also satisfy the $\mathfrak{gl}_n$ algebra relations, thus $\check r$ (\ref{special1}) is $\mathfrak{gl}_n$ symmetric. In this particular case, as is clear from the computation above, two invertible linear maps are involved, $F_i\in \mbox{End}({\mathbb C}^n \otimes {\mathbb C}^n), i \in \{ 1, 2\},$ such that $F_1 := u \otimes 1_V $ and $F_2 := 1_V  \otimes u^{-1}$ and $F_i\Delta(e_{x,y})F_i^{-1} = \Delta_i(e_{x,y}).$
\end{proof}

By iteration one derives the $N$ co-products:
$ \Delta_1^{(N)} =  ( \Delta_1^{(N-1)}  \otimes \mbox{id})  \Delta_1$ and  
$\Delta_2^{(N)} = (\mbox{id} \otimes \Delta_2^{(N-1)}) \Delta_2$,
which explicitly read as
\begin{eqnarray}
&& \Delta_1^{(N)}(e_{x,y}) =\sum_{k=1}^N 1_V  \otimes \ldots \otimes e_{\sigma^{N-k}(x),\sigma^{N-k}(y)}
\otimes \ldots \otimes 1_V , \label{delta1}\\
&& \Delta_2^{(N)}(e_{x,y}) =\sum_{k=1}^N 1_V \otimes \ldots \otimes e_{\tau^{k-1}(x),\tau^{k-1}(y)}
\otimes \ldots \otimes 1_V. \label{delta2}
\end{eqnarray}
The above expressions can be written in a compact form as:
$\Delta^{(N)}_i(e_{x,y}) = F_i^{(N)}\Delta^{(N)}(e_{x,y})(F_i^{(N)})^{-1},$ where 
$\Delta^{(N)}(e_{x,y})=  \sum_{k=1}^N \mbox{id} \otimes \ldots \otimes\underbrace{e_{x,y}}_{k^{th}\  \mbox{position}} \otimes \ldots \otimes \mbox{id},$ and we define 
$F_{1}^{(N)}:= u^{N-1} \otimes u^{N-2} \otimes \ldots \otimes u \otimes 1_V$  and 
${F}_2^{(N)} := 1_V \otimes u^{-1} \otimes u^{-2}\otimes \ldots \otimes u^{-(N-1)}$ (see also relevant findings in 
\cite{Doikoutw, DoGhVl} ),
where $F_i^{(2)}=: F_i,$ $i\in \{1,2\}.$ Notice that co-associativity does not hold in this case. In fact, it turns out that the quantum algebra associated to generic set-theoretic solutions
is a quasi-bialgebra (see e.g. \cite{Doikoutw, DoGhVl}).

We note that local open  Hamiltonians for generic Baxterized solutions coming from involutive solutions were systematically 
derived in \cite{DoiSmo1}. Specifically, the open Hamiltonian associated to Lyubashenko's solution with special boundary conditions is givens as
\begin{equation}
{\cal H}_c = \sum_{j=1}^{N-1}\sum_{a=1}^n e^{(j)}_{a,b+c} e^{(j+1)}_{b, a-c}, 
\end{equation}
$ c \in \{1,2,\ldots, n-1\}.$
This is a $\mathfrak{gl}_n$ symmetric Hamiltonian, i.e. $\big[{\cal H}_c,\ \Delta_i^{(N)}(y)\big ]=0,$ $y \in \mathfrak{gl}_n,$ $i\in \{1,2\},$ as opposed to the periodic one, which is not $\mathfrak{gl}_n$ symmetric. Recall $\Delta^{(N)}_i$ are defined in (\ref{delta1}),(\ref{delta2}).

Motivated by the case of Lyubashenko's solution and the suitable twist of $\mathfrak{gl}_n$ (see also \cite{DoGhVl}) we present in the subsequent section the full theory of quantum algebras associated to set-theoretic solutions. These are new types of quasi-triangular Hopf algebras introduced in \cite{Doikoutw, DoRy22, DoRySt}. The quantum algebras associated to the set-theoretic solutions of the parametric Yang-Baxter equation were also introduced in \cite{Doikoup}.

\section{Set-theoretic quasi-triangular Hopf algebras $\&$ Drinfel'd twists} \label{sec:5}

\noindent In this section we discuss the Yang-Baxter algebras associated to rack type and set-theoretic solutions 
as quasi-triangular Hopf algebras \cite{Doikoutw, DoRySt}. 
The Hopf algebra theory associated to set-theoretic solutions has been developed in \cite{Doikoutw, DoRy22, DoRySt}
(see also \cite{EtScSo99} and  \cite{Andru} in connection with pointed Hopf algebras and \cite{Lebed, braceh} on Hopf algebras in 
connection to braces \cite{Ru05}--\cite{Ru19}, \cite{GuaVen}).   
The associated universal ${\cal R}$-matrices are also derived in this frame. 
A full analysis of set-theoretic twists in the Yangian $\mathfrak{gl}_n$ is presented in \cite{Doikoutw, DoGhVl, DoikouYang}.

We start our analysis with the rack and quandle algebras and the construction of the associated universal 
${\cal R}$-matrix.
We then extend the algebra to the {\it decorated rack algebra} and via a suitable admissible Drinfel'd twist 
we construct the corresponding universal ${\cal R}-$matrix. 
For a more detailed exposition on the proofs presented in this section the interested reader is referred to \cite{Doikoutw, DoRySt}.

\subsection{The rack \& quandle  algebras}

\noindent We first define the {\it rack} 
and {\it quandle} algebras \cite{DoRySt}.

\begin{defn} \label{def1} (The rack algebra)
Let $(X, \triangleright)$ be a finite magma, or such that $a\triangleright$ is surjective, for every $a\in X$.
We say that the unital, associative algebra ${\cal A},$ over a field $k$ 
generated by indeterminates ${1_\cal A}$ (the unit element), $~q_a, \ q^{-1}_a,\ h_a\in {\cal A}$ 
($h_a = h_b \Leftrightarrow a =b$)
and relations, $a,b \in X:$
\begin{eqnarray}
q_a^{-1} q_a= q_a q_a^{-1} = 1_{\cal A},\quad q_a q_b = q_b  q_{b \triangleright a}, \quad  h_a  h_b =\delta_{a, b} h_a, \quad q_b  h_{b\triangleright a} = h_a q_b, \label{qualg}
\end{eqnarray}
is a {\it rack} algebra.
\end{defn}

\begin{defn}  \label{qualgd} (The quandle algebra)
A rack algebra is called a quandle algebra if there is a magma $(X, \bullet),$ 
such that for all $a,b \in X,$ $a\bullet : X \to X$ is a bijection and $a\bullet b = b\bullet (b \triangleright a).$
\end{defn}

The following proposition fully justifies the names rack and quandle algebras.
\begin{pro} \label{proro1}
\label{qua1} Let ${\cal A}$ be the rack algebra,
then for all $a,b,c\in X,$ $c \triangleright (b \triangleright a) = (c \triangleright b)\triangleright (c \triangleright a),$ and $a\triangleright $ is bijective, i.e. $(X, \triangleright)$ is a rack. If ${\cal A}$ is the quandle algebra, then in addition, for all $a\in X,$ $a \triangleright a =a,$ i.e. $(X, \triangleright )$ is a quandle.
\end{pro}
\begin{proof}
We compute $h_a q_b q_c$ using the associativity of the rack algebra and due to invertibility of $q_a$ for all $a \in X$ we conclude for all $a,b,c \in X$, 
\[h _{c \triangleright(b \triangleright a)} =h_{(c \triangleright b) \triangleright ( c\triangleright a)}\  
 \Rightarrow\ c \triangleright(b \triangleright a) = (c \triangleright b) \triangleright ( c\triangleright a).\]
 
We assume $c\triangleright a = c\triangleright b$, then $q_c h_{c\triangleright a} = q_c h_{c\triangleright b}$,  by the fourth relation in \ref{qualg}, we get $h_aq_c = h_bq_c$ and by the invertibility of $q_c$, $h_a = h_b,$ hence $a=b$, i.e. $a\triangleright$ is bijective for all $a\in X$ and thus $(X, \triangleright)$ is a rack.

Moreover,  we recall for the quandle algebra $a\bullet b = b \bullet (b \triangleright a),$ then for $a=b$ and 
by recalling bijectivity and hence
left cancellation in $(X, \bullet)$ we conclude that $a\triangleright a =a$ for all $a \in X,$ i.e. $(X, \triangleright )$ is a quandle.
\end{proof}

\begin{lemma} \label{lemmac}
Let $c= \sum_{a\in X} h_a,$ then $c$ is a central element of the rack algebra ${\cal A}$. \end{lemma}
\begin{proof} The proof is direct by means of the definition of the algebra ${\cal A}$ and Proposition  \ref{proro1}. Without loss of generality let $c=1_{\cal A}.$ 
\end{proof}

Having defined the rack algebra we are now 
in the position to identify the associated universal 
${\cal R}$-matrix (solution of the Yang-Baxter equation).
\begin{pro} 
Let ${\cal A}$ be the rack algebra  and ${\cal R} \in {\cal A} \otimes {\cal A}$ 
be an invertible element, such that ${\cal R} = \sum_{a\in X} h_a \otimes q_a.$
Then ${\cal R}$ satisfies the Yang-Baxter equation
\begin{equation}
{\cal R}_{12} {\cal R}_{13} {\cal R}_{23} = {\cal R}_{23} {\cal R}_{13} {\cal R}_{12},  \nonumber
\end{equation}
where ${\cal R}_{12} = \sum_{a\in X} h_a \otimes q_a \otimes 1_{\cal A}, $ ${\cal R}_{13} = \sum_{a\in X} h_a  
\otimes 1_{\cal A} \otimes q_a,$ and  ${\cal R}_{23} = \sum_{a\in X} 1_{\cal A} \otimes h_a \otimes q_a.$
\end{pro}
\begin{proof}
The proof is  a direct computation of the two sides of the Yang-Baxter equation 
(and use of the fundamental relations (\ref{qualg})):
\begin{eqnarray}
&& \mbox{LHS}: \quad   \sum_{a,b,c \in X} h_a h_b \otimes  q_a  h_c \otimes q_b q_c =  
\sum_{a,b,c \in X} h_a \otimes  q_a  h_c \otimes q_a q_c = \sum_{a,b,c \in X} h_a \otimes  
q_a  h_{a\triangleright c} \otimes q_a q_{a\triangleright c}\nonumber\\
&& \mbox{RHS}: \quad  \sum_{a,b,c\in X}   h_b h_a \otimes  h_c  q_a \otimes q_c q_b=  
\sum_{a,b,c\in X}    h_a \otimes   q_a h_{a\triangleright c} \otimes q_c q_a, \nonumber
\end{eqnarray}
where we have used that $a\triangleright$ is bijective. Then due to the basic relation $q_a q_b = q_b  q_{b \triangleright a},$ we show that LHS$=$RHS,
and this concludes our proof.
\end{proof}

\begin{rem} \label{remc}
The universal ${\cal R}-$matrix is invertible, indeed from Lemma \ref{lemmac},
$\sum_{a\in X} h_a= 1_{\cal A},$ hence
${\cal R}^{-1} = \sum_{a\in X} h_a \otimes q_a^{-1}.$  
\end{rem}

\begin{rem} \label{remfu} 
{\bf Fundamental representation:} Let ${\cal A}$ be the rack algebra and $\rho: {\cal A} \to \mbox{End}(V)$ be the map defined by
\begin{equation}
q_a \mapsto \sum_{x \in X} e_{x, a \triangleright x}, \quad h_a\mapsto e_{a,a}, \label{remfu1}
\end{equation}
where $(X, \triangleright)$ is a rack.
Then ${\cal R} \mapsto r:= 
\sum_{a,b\in X} e_{b,b} \otimes e_{a, b\triangleright a},$  is the linearized version of a rack solution of the Yang-Baxter equation.
We note that $r$ is invertible, because $a\triangleright: X \to X$ is a bijection for all $a \in X$, then $r^{-1} = \sum_{a,b\in X} e_{b,b} \otimes e_{b\triangleright a, a}.$ 

Let ${\cal P} = \sum_{a,b \in X} e_{a,b} \otimes e_{b,a}$ be the permutation (flip) operator, 
then the solution of the braid equation is the linearized version of rack solution,
 $\check r = {\cal P} r = \sum_{a,b} e_{a,b} \otimes e_{b, b \triangleright a}.$ 

Note that in the special case where $\check r$ is involutive, i.e. $\check r^2 = \mbox{id},$ then $a\triangleright b = b,$ for all $a,b \in X,$
which means that the rack algebra (Definition \ref{def1}) becomes a commutative algebra. In this case $\check r = \sum_{a,b \in X} e_{a,b} \otimes e_{b,a},$ i.e. $\check r$ reduces to the permutation operator, whereas $r$ reduces to the identity
\end{rem}

\begin{thm}\label{basica1}  
Let  ${\cal A}$ 
be the quandle algebra (Definition \ref{qualgd}).  Let also
${\cal R}= \sum_{a\in X} h_a \otimes q_a$ and for all $a,b \in X,$  $q_a q_b = q_{a\bullet b}$. 
If $(X, \bullet, e)$ is a group, then 
$({\cal A}, \Delta, \epsilon, S,  {\cal R})$ is a quasi-triangular Hopf algebra:
\begin{itemize}
\item Co-product. $\Delta: {\cal A} \to {\cal A} \otimes {\cal A},$
$~\Delta(q_a^{\pm 1}) = q_a^{\pm 1}\otimes q_a^{\pm 1}$
and $\Delta(h_a) = \sum_{b,c \in X} 
h_b \otimes h_c\big |_{b\bullet c = a}.$
\item Co-unit. $~\epsilon: {\cal A} \to k,$ $\epsilon(q_a^{\pm 1}) = 1,$ $~\epsilon(h_a) =  \delta_{a,e}.$ 
\item Antipode. $~S: {\cal A} \to {\cal A},$ $~S(q_a^{\pm 1}) =q_a^{\mp 1},$ $S(h_a)= h_{a^*},$ where $a^*$ 
is the inverse in $(X, \bullet)$ for all $a \in X.$
\end{itemize}
\end{thm}
\begin{proof} It is straightforward to check that
$\Delta$ 
is an algebra homomorphism. Indeed, this can be explicitly checked via the distributivity condition $a\triangleright(b \bullet c) = (a \triangleright b) \bullet (a \triangleright c),$ which readily follows from, $a\triangleright b = a^{*}\bullet b \bullet a.$

We are now going to prove that all the axioms of a quasi-triangular Hopf algebra hold.

Given the co-products of the generators we have to check co-associativity and 
also uniquely derive the counit $\epsilon: {\cal A} \to k$ (homomorphism) 
and antipode $S: {\cal A} \to {\cal A}$ (anti-homomorphism).
\begin{enumerate}[{(i)}]
\item Co-associativity.: 
\begin{eqnarray}
&& (\mbox{id} \otimes \Delta)\Delta(q_a) =  (\Delta \otimes \mbox{id}) \Delta(q_a) =  
q_a \otimes q_a \otimes q_a, \quad  \nonumber\\
&& (\mbox{id} \otimes \Delta)\Delta(h_a) =  (\Delta \otimes \mbox{id}) \Delta(h_a) = 
\sum_{b,c,d\in X} h_b \otimes h_c \otimes h_d 
\big |_{b\bullet c \bullet d =a}. \nonumber 
\end{eqnarray}
\item Counit: $ (\epsilon \otimes \mbox{id})\Delta(x) = (\mbox{id} \otimes \epsilon)
\Delta(x) =x,$ for all $x\in  \{q_a,\ q_a^{-1}, h_a\}.$\\
The generators $q_a$ are group-like elements, so $\epsilon(q_a) = 1,$ and 
\begin{equation}
\sum_{a,b\in X} \epsilon(h_a) h_b = \sum_{a,b} h_a \epsilon(h_b) 
\big |_{a\bullet b =c} = h_c \Rightarrow \epsilon(h_a) = \delta_{a,e}.
\end{equation}
\item Antipode: $m\big ((S \otimes \mbox{id})\Delta(x)) =  
m\big ((\mbox{id} \otimes S )\Delta(x)) = \epsilon(x)1_{\cal A}$ for all $x\in \{q_a,\ q_a^{-1}, h_a\}.$\\
For $q_a,$ we immediately have  $S(q_a) = q_a^{-1}$ and (recall $h_{a} h_{b} = \delta_{a,b} h_a$ 
and $\sum_{a\in X} h_a = 1_{\cal A}$)
\begin{equation}
\sum_{a,b \in X} S(h_a) h_b \big |_{a\bullet b =c} = \sum_{a,b \in X} h_a S(h_b) 
\big |_{a\bullet b =c} = \delta_{c, e} 1_{\cal A}\ \Rightarrow\ S(h_a) = h_{a^{*}},  \nonumber\\
\end{equation}
where $a^*$ is the inverse in $(X, \bullet)$ for all $a \in X.$
So $({\cal A}, \Delta, \epsilon, S)$ is a Hopf algebra.

\item Moreover,
\begin{eqnarray}
&& {\cal R}_{13} {\cal R}_{12} =\sum_{a\in X}  h_a \otimes q_a \otimes 
q_a=\sum_{a\in X} h_a \otimes \Delta(q_a) = (\mbox{id} \otimes \Delta){\cal R},  \label{v1}\\
&& {\cal R}_{13} {\cal R}_{23} =\sum_{a,b\in X} h_a \otimes h_b \otimes q_c\big |_{a\bullet b =c} 
=  \sum_{c\in X} \Delta(h_c) \otimes  q_c  = (\Delta \otimes \mbox{id}){\cal R}. \label{v2}
\end{eqnarray}

\item It is also readily shown from the relations of the rack algebra that
\begin{eqnarray}
  \Delta^{(op)}(q_a) {\cal R} = {\cal R} \Delta(q_a)\quad  \Delta^{(op)}(h_a) {\cal R} = {\cal R} \Delta(h_a),  \label{comm}
\end{eqnarray}
where $\Delta^{(op)} = \pi \circ \Delta,$ $\pi$ is the flip map.

We conclude that $({\cal A}, \Delta, \epsilon, S, {\cal R})$ is a quasi-triangular Hopf algebra.
\hfill \qedhere
\end{enumerate}
\end{proof}

\subsection{The set-theoretic YB algebras}

\noindent In this subsection we suitably extend the rack and quandle algebras in order to construct 
the universal ${\cal R}-$matrix associated to general set-theoretic solutions of the Yang-Baxter equation.

We first define the {\it decorated rack} algebra and the set-theoretic Yang-Baxter algebra.

\begin{defn}  \label{setalgd1} (The decorated rack algebra.) Let ${\cal A}$ 
be the rack algebra (Definition \ref{qualg}). Let also $\sigma_a, \ \tau_b: X\to X,$ and $\sigma_a$ be 
bijective for all $a\in X$. We say that the unital, associative algebra $\hat {\cal A}$ over $k,$
generated by indeterminates  $1_{\hat {\cal A}} ,q_a, q_a^{-1}, h_a \in {\cal A}$ ($h_a = h_b \Leftrightarrow a =b$) and 
$w_a, w^{-1}_a\in \hat {\cal A},$ $a \in X,$ 
and relations, $a,b \in X:$
\begin{eqnarray}
&& q_a^{-1} q_a = q_aq_a^{-1} =1_{\hat {\cal A}}, ~~ q_a q_b = q_b 
q_{b \triangleright a}, ~~  h_a  h_b =\delta_{a,b} h_a, ~~ q_b  h_{b\triangleright a} = h_a q_b,\nonumber\\
&&w_a^{-1} w_a =w_aw_a^{-1} =1_{\hat {\cal A}}, ~~  w_a w_b = 
w_{\sigma_a(b)} w_{\tau_{b}(a)} ~~  w_a h_b = h_{\sigma_a(b)} w_a, 
~~ w_a q_b = q_{\sigma_a(b)} w_a \label{qualgbb}
\end{eqnarray}
is a {\it decorated rack algebra}.
\end{defn}

\begin{lemma} \label{lemmad}
Let $c= \sum_{a\in X} h_a,$ then $c$ is a central element of the decorated rack 
 algebra $\hat {\cal A}$.
\end{lemma}
\begin{proof} The proof is straightforward by means of the definition of the algebra $\hat {\cal A}.$ 
We consider henceforth, without loss of generality,  $c=1_{\hat {\cal A}}$ (see also Lemma \ref{lemmac}). 
\end{proof}

\begin{pro} 
\label{qua2} 
Let $\hat {\cal A}$ be the decorated rack algebra,
then for all $a,b,c \in X,$
\begin{eqnarray}
&&  \sigma_{a}(\sigma_b(c))= \sigma_{\sigma_{a}(b)}(\sigma_{\tau_{b}(a)}(c))  \quad \& \quad   
\sigma_c(b) \triangleright \sigma_{c}(a) = \sigma_c(b \triangleright a). \nonumber
\end{eqnarray}
\end{pro}
\begin{proof}
We compute $w_a w_b h_c$ using the associativity of the algebra and the invertibility of $w_a$ for all $a\in X$ and we deduce for all $a,b,c \in X,$
\begin{equation}
h_{\sigma_{\sigma_{a}(b)}(\sigma_{\tau_{b}(a)}(c))} = h_{\sigma_a(\sigma_b(c))}\  
\Rightarrow \ \sigma_{\sigma_{a}(b)}(\sigma_{\tau_{b}(a)}(c))= \sigma_a(\sigma_b(c)).  \label{basiko} \nonumber
\end{equation}
\\
We also compute $h_aq_b w_c$, via associativity and the invertibility of $q_a,\ w_a$ for all $a\in X,$ we obtain for all $a,b,c \in X,$
\begin{equation}
h_{\sigma_c^{-1}(b)\triangleright \sigma_{c}^{-1}(a)} = h_{\sigma^{-1}_c(b \triangleright a)}\ \Rightarrow \  
\sigma_c^{-1}(b)\triangleright \sigma_{c}^{-1}(a) = \sigma^{-1}_c(b \triangleright a),\nonumber
\end{equation}
from the latter it immediately follows, $ \sigma_c(b)\triangleright \sigma_{c}(a) = \sigma_c(b \triangleright a).$
\end{proof}

\begin{defn}  
\label{setalgd} 
(The set-theoretic Yang-Baxter algebra.) 
Let ${\cal A}$ 
be the quandle algebra. Let also $\sigma_a, \ \tau_b: X\to X,$ and $\sigma_a$ be bijective for all $a\in X.$  We say that the unital, associative algebra $\hat {\cal A}$ over $k,$
generated by indeterminates  $1_{\hat {\cal A}} ,q_a, q_a^{-1}, h_a \in {\cal A}$ ($h_a = h_b \Leftrightarrow a =b$) and 
$w_a, w^{-1}_a\in \hat {\cal A},$ $a \in X,$ 
and relations, (\ref{qualgbb}) is a set-theoretic  Yang-Baxter algebra.
\end{defn}

\begin{pro} \label{basica2b} (Hopf algebra) 
Let $\hat {\cal A}$ 
be the set-theoretic Yang-Baxter algebra,
 ${\cal R} = \sum_{b\in X} h_b\otimes q_b$  and $({\cal A}, \Delta,  \epsilon, S,  {\cal R})$  
be the quasi-triangular Hopf algebra of Theorem \ref{basica1}. If for all $a,b,x \in X,$ 
\begin{equation}
 \sigma_x(a) \bullet \sigma_x(b) = \sigma_x(a\bullet b), \label{condition0}
\end{equation} 
then, 
\begin{enumerate}
\item $(\hat {\cal A}, \Delta, \epsilon, S)$ is a Hopf algebra with $\Delta(w_a) = w_a\otimes w_a,$ 
for all $a \in X.$
\item  $\Delta(w_a) {\cal R} = {\cal R} \Delta(w_a),$ for all $a \in X.$
\end{enumerate}
\end{pro}
\begin{proof}  
In our proof below we are using the Definition \ref{setalgd} and (\ref{condition0}).
\begin{enumerate}
\item The coproduct $\Delta$ is an algebra homomorphism.  It is sufficient to check below the 
consistency of all algebraic relations of Definition \ref{setalgd} for the corresponding 
coproducts and (\ref{condition0}).
Then we have for $Y_b \in \{h_b,\ q_b\}$ and for all $a, b\in X,$
\begin{equation}
\Delta(w_a) \Delta(w_b) = \Delta(w_{\sigma_a(b)}) \Delta(w_{\tau_b(a)}),  \quad   
\Delta(w_a) \Delta(Y_b) = \Delta(Y_{\sigma_a(b)}) \Delta(w_a). \nonumber
\end{equation}
Also, $w_a$ is a group-like element, thus the counit and antipode are given as: 
$\epsilon(w_a) =1$ and $S(w_a) = w^{-1}_a.$ Recall that the coproducts, counits 
and antipodes of the generators $h_a,\ q_a$ are given in
Theorem \ref{basica1}.

\item By a direct computation and using the algebraic relations of the 
Definition \ref{setalgd} we conclude for all $a \in X,$
$~\Delta(w_a) {\cal R} = {\cal R} \Delta(w_a).$
\hfill \qedhere
\end{enumerate}
\end{proof}

\subsection{Set-theoretic Drinfel'd twist}
In this subsection we introduce the universal set-theoretic (or combinatorial) Drinfel'd twist (\cite{Doikoutw, DoGhVl, DoRySt} (see also, relevant construction in \cite{Sol, LebVen}). Using the twist, we will be able to obtain the universal ${\cal R}$-matrix associated with the set-theoretic Yang-Baxter (YB) algebra.

Before we introduce the set-theoretic twist, we recall a general statement \cite{Drinfeld}.
\begin{pro} \label{prot} (Admissible Drinfel'd twist \cite{Drinfeld}) Let ${\cal A}$ be a unital, associative algebra, ${\cal F}, {\cal R} \in {\cal A} \otimes {\cal A}$ be invertible elements and ${\cal R}$ satisfies the Yang-Baxter equation. Let also ${\cal F}_{1,23}, {\cal F}_{12,3} \in {\cal A}^{\otimes 3},$ such that
\begin{enumerate}
\item  ${\cal F}_{23}{\cal F}_{1,23} = {\cal F}_{12} {\cal F}_{12,3},$ where recall ${\cal F}_{12} = {\cal F} \otimes 1_{\cal A}$ and ${\cal F}_{23} = 1_{\cal A} \otimes {\cal F}.$ 
\item  ${\cal F}_{1,32}  {\cal R}_{23} = {\cal R}_{23} {\cal F}_{1,23}$ and ${\cal F}_{21,3} {\cal R}_{12} = {\cal R}_{12} {\cal F}_{12,3}.$
\end{enumerate}
That is, ${\cal F}$ is an admissible Drinfel'd twist. Define also ${\cal R}^F := {\cal F}^{(op)} {\cal R} {\cal F},$ ${\cal F}^{(op)} = \pi({\cal F})$ where $\pi: {\cal A} \otimes {\cal A} \to {\cal A} \otimes {\cal A}$ is the flip map. Then ${\cal R}^F$ also satisfies the Yang-Baxter equation. 
\end{pro}
\begin{proof}

It is convenient to introduce some handy notation that can be used in the following.
First, let ${\cal F}_{123} := {\cal F}_{12}{\cal F}_{1,23} = {\cal F}_{23} {\cal F}_{1,23}.$
Let also $i,j,k \in \{1,2,3\},$ then ${\cal F}_{jik} = \pi_{ij}( {\cal F}_{ijk})$ and
${\cal F}_{ikj} = \pi_{jk}({\cal F}_{ijk}),$ where $\pi$ is the flip map. This notation 
describes all possible permutations of the indices $1,\ 2,\ 3$. 

The proof is quite straightforward, \cite{Drinfeld},
we just give a brief outline here. 
We first prove that ${\cal F}_{jik}{\cal R}_{ij}{\cal F}^{-1}_{ijk} = {\cal R}^F_{ij},$ indeed via condition (2) of the proposition the definition of ${\cal R}^{F}$ and the notation introduced above we have
\begin{equation}
{\cal F}_{jik}{\cal R}_{ij}{\cal F}^{-1}_{ijk} = {\cal F}_{ji} {\cal F}_{ji,k} {\cal R}_{ij} {\cal F}_{ijk} = {\cal F}_{ji}{\cal R}_{ij}{\cal F}_{ij,k}{\cal F}_{ijk}^{-1} = {\cal R}^{F}_{ij} {\cal F}_{ij}{\cal F}_{ij,k}{\cal F}^{-1}_{ijk} = {\cal R}^F_{ij}.
\end{equation}
Similarly, it is shown that ${\cal F}_{ikj}{\cal R}_{jk}{\cal F}^{-1}_{ijk} = {\cal R}^F_{jk}.$

Then from the YBE we have,
\begin{equation}
{\cal  F}_{321} {\cal R}_{12} {\cal R}_{13} {\cal R}_{23} ={\cal  F}_{321} 
{\cal R}_{23}{\cal R}_{13}{\cal R}_{12}\ \Rightarrow\  {\cal R}^F_{12} {\cal R}^F_{13} {\cal R}^F_{23}{\cal  F}_{123} 
=  {\cal R}^F_{23}{\cal R}^F_{13}{\cal R}^F_{12}{\cal  F}_{123}. \nonumber
\end{equation}
But ${\cal F}_{123}$ is invertible, hence ${\cal R}^F$ indeed satisfies the Yang-Baxter equation.
\end{proof}

Below we prove the main theorem on the set-theoretic Drinfel'd twist (see more details on the proof in \cite{Doikoutw, DoRySt}).
\begin{thm} \label{twist2} (Set-theoretic Drinfel'd twist \cite{Doikoutw, DoRySt}) 
Let ${\cal R} = \sum_{a\in X} h_a \otimes q_a\in {\cal A} \otimes {\cal A}$ 
be the universal rack ${\cal R}-$matrix. Let also $\hat {\cal A}$
be the decorated rack algebra, ${\cal F}\in \hat {\cal A} \otimes \hat  {\cal A},$  
such that ${\cal F} =\sum_{b\in X}h_b \otimes w_b^{-1}$ and
${\cal R}_{ij}^F := {\cal F}_{ji} {\cal R}_{ij} {\cal F}_{ij}^{-1},$ $i,j \in \{1,2,3\}.$
We also define:
\begin{eqnarray}
{\cal F}_{1,23} := \sum_{a\in X} h_a\otimes w_a^{-1} \otimes 
w_a^{-1}=,  \quad {\cal F}^*_{12,3} : =  \sum_{a,b \in X}h_a\otimes h_{\sigma_a(b)} 
\otimes w^{-1}_{b} w^{-1}_a.
\end{eqnarray}
Let also  for every
$a,b \in X,$ $~b\triangleright a = \sigma_b(\tau_{\sigma^{-1}_a(b)}(a)).$ 
Then, the following statements are true:
\begin{enumerate}
\item ${\cal F}_{12} {\cal F}^*_{12,3} ={\cal F}_{23} {\cal F}_{1,23} =: {\cal F}_{123}.$
\item For  $i,j,k \in \{1,2,3\}$:
(i) ${\cal F}_{ikj} {\cal R}_{jk} ={\cal R}_{jk}^F {\cal F}_{ijk}$
and 
(ii)  ${\cal F}_{jik} {\cal R}_{ij} ={\cal R}_{ij}^F {\cal F}_{ijk}.$
\end{enumerate}
That is,  ${\cal F}$ is an  admissible Drinfel'd twist.
\end{thm}
\begin{proof}
The proof is straightforward based on the underlying algebra $\hat {\cal A}.$
\begin{enumerate}
\item Indeed, this is proved by a direct computation and use of the decorated rack algebra. 
In fact, ${\cal F}_{123} = \sum_{a,b\in X} h_a \otimes h_b w_a^{-1} \otimes w_b^{-1} w_a^{-1}.$
\item Given the notation introduced in the proof of Proposition \ref{prot} it  suffices to show  that
${\cal F}_{132} {\cal R}_{23} = {\cal R}^F_{23} {\cal F}_{123}$ and ${\cal F}_{213} 
{\cal R}_{12} = {\cal R}^F_{12} {\cal F}_{123}.$ \\
\noindent (i) 
Due to the fact that for all $a \in X,$
$ \Delta(w_a){\cal R} =  {\cal R} \Delta(w_a)$ 
(see Proposition \ref{basica2b})
we arrive at ${\cal F}_{1,32} {\cal R}_{23} = {\cal R}_{23} F_{1,23},$ then
\[{\cal F}_{132} {\cal R}_{23} = {\cal F}_{32} {\cal F}_{1,32} {\cal R}_{23} ={\cal F}_{32} 
{\cal R}_{23} {\cal F}_{1,23} = {\cal R}_{23}^F{\cal F}_{123}. \]
\noindent (ii) By means of the relations 
of the decorated rack algebra $\hat {\cal A}$ we compute:
\[ {\cal F}^*_{21,3}{\cal R}_{12} = \sum_{a,c\in X} h_a \otimes q_a h_{a\triangleright c} \otimes 
 (w_c  w_{\sigma^{-1}_c(a)} )^{-1},\quad {\cal R}_{12}{ \cal F}^*_{12,3}= \sum_{a,b \in X} 
h_a \otimes q_a h _{\sigma_a(b)} \otimes (w_aw_b)^{-1}.\] 
Due to the fact that $~b\triangleright a = \sigma_b(\tau_{\sigma^{-1}_a(b)}(a))$ and $w_a w_b = w_{\sigma_a(b)} w_{ \tau_b(a)}$
we conclude that ${\cal F}^*_{21,3}{\cal R}_{12} = 
{\cal R}_{12}{ \cal F}^*_{12,3}$ and consequently (recall ${\cal F}_{213} = {\cal F}_{21} {\cal F}^*_{21,3}$)
\[{\cal F}_{213} {\cal R}_{12} = {\cal F}_{21} {\cal F}^*_{21,3} {\cal R}_{12} 
={\cal F}_{21} {\cal R} _{12}{\cal F}^*_{12,3} = {\cal R}_{12} ^F{\cal F}_{123}. \hfill \qedhere
\]
\end{enumerate}
\end{proof}

\begin{cor} \label{twist22}  
 Let $\hat {\cal A}$
be the decorated rack algebra, ${\cal F}\in \hat {\cal A} \otimes \hat  {\cal A},$  
such that ${\cal F} =\sum_{b\in X}h_b \otimes w_b^{-1}$ and
${\cal R}_{ij}^F := {\cal F}_{ji} {\cal F}_{ij}^{-1},$ $i,j \in \{1,2,3\}.$
We also define:
\begin{eqnarray}
{\cal F}_{1,23} := \sum_{a\in X} h_a\otimes w_a^{-1} \otimes 
w_a^{-1}=,  \quad {\cal F}^*_{12,3} : =  \sum_{a,b \in X}h_a\otimes h_{\sigma_a(b)} 
\otimes w^{-1}_{b} w^{-1}_a.
\end{eqnarray}
Let also  for every
$a,b \in X,$ $\sigma_{\sigma_a(b)}(\tau_b(a)) =a.$ 
Then, the following statements are true:
\begin{enumerate}
\item ${\cal F}_{12} {\cal F}^*_{12,3} ={\cal F}_{23} {\cal F}_{1,23} =: {\cal F}_{123}.$
\item For  $i,j,k \in \{1,2,3\}$:
(i) ${\cal F}_{ikj}  ={\cal R}_{jk}^F {\cal F}_{ijk}$
and 
(ii)  ${\cal F}_{jik}  ={\cal R}_{ij}^F {\cal F}_{ijk}.$
\item ${\cal R}_{12}^F {\cal R}_{21}^F = 1_{\hat {\cal A}^{\otimes 2}},$ i.e. ${\cal R^F}$ is reversible.
\end{enumerate}
\end{cor}
\begin{proof}
   This is an immediate consequence of Theorem \ref{twist2}. ${\cal R}^F$ is reversible by construction.
\end{proof}

We are now going to examine the twisted ${\cal R}-$matrix as well 
as the twisted co-products (see also \cite{Doikoutw, DoRy22}). 
\begin{rem} 
(Twisted universal ${\cal R}$-matrix )We extract now the explicit 
expressions of the twisted universal ${\cal R}-$matrix and the twisted coproducts of the algebra. 
We recall the admissible twist ${\cal F} = \sum_{b\in X}h_b \otimes w_b^{-1}.$

\begin{itemize}

\item The twisted ${\cal R}-$matrix: 
\[{\cal R}^F = {\cal F}^{(op)} {\cal R} {\cal F}^{-1} = \sum_{a,b\in X}h_bw_a^{-1}  
 \otimes h_aq_{\sigma_a(b)} w_{\sigma_a(b)}.\]

\item The twisted coproducts: $\Delta_F(y) = {\cal F} \Delta(y) {\cal F^{-1}},$ $y\in \hat {\cal A}$  (in this Remark $\hat {\cal A}$ 
denotes the set-theoretic Yang-Baxter algebra) and we recall, for $a\in X,$
\[ \Delta(w_a) =w_a\otimes w_a,\quad \Delta(h_a) = 
\sum_{b,c\in X} h_b \otimes h_c\big |_{b\bullet c =a}, \quad \Delta(q_a) =q_a \otimes q_a.\]
Then, the twisted coproducts read as:
$~\Delta_F(w_a) = \sum_{b \in X} h_{\sigma_a(b)} w_a \otimes w _{\tau_b(a)},$\\
$\Delta_F(h_a) = \sum_{b \in X} h_b \otimes w_b^{-1} h_c w_b\big |_{b\bullet c =a},$ $~\Delta_F(q_a) =\sum_{b \in X} q_a h_{a\triangleright b} \otimes  w_b^{-1} q_a 
w_{a\triangleright b},$
\end{itemize}
and it immediately follows that ${\cal R}^F \Delta_F(Y) = \Delta_F^{(op)}(Y) {\cal R}^F,$ $Y \in \hat {\cal A}.$
\end{rem}

\begin{rem} \label{remfu2b}  
{\bf Fundamental representation $\&$ the set-theoretic solution:}\\ 
Let $\rho: \hat  {\cal A} \to \mbox{End}(V),$ such that
\begin{equation}
q_a \mapsto \sum_{x \in X} e_{x, a \triangleright x}, \quad h_a\mapsto e_{a,a}, \quad  
w_a \mapsto \sum_{b \in X} e_{\sigma_a(b),b},\label{repbb}
\end{equation}
then ${\cal F} \mapsto F := \sum_{a,b \in X} e_{a,a} \otimes e_{\sigma_a(b),b}$ and ${\cal R}^F \mapsto r^F:= 
\sum_{a,b\in X} e_{b,\sigma_a(b)} \otimes e_{a, \tau_b(a)},$
where we recall that for all $a,b,c \in X,$ $\sigma_{\sigma_{a}(b)}(\sigma_{\tau_{b}(a)}(c)) = \sigma_a(\sigma_b(c)$ and $\tau_{b}(a):=\sigma_{\sigma_a(b)}^{-1}(\sigma_a(b) 
\triangleright a)$ and $(X, \triangleright)$ is a rack (see also \cite{Sol, LebVen, DoRy22}). This is the linearized version of the general set-theoretic 
solution of the Yang-Baxter equation.

Recall, ${\cal P} = \sum_{a,b \in X} e_{a,b} \otimes e_{b,a}$ is the permutation (flip) operator, 
then the solution of the braid equation is the linearized set-theoretic solution,
 $\check r^F = {\cal P} r^F = \sum_{a,b\in X} e_{a,\sigma_a(b)} \otimes e_{b, \tau_b(a)}.$

 In the special case, where the set-theoretic solution of the braid equation is invertible, i.e. $(\check r^{F})^2 =\mbox{id},$ then $\sigma_{\sigma_a(b)}(\tau_b(a)) = a,$ which leads to $a\triangleright b = b$ for all $a,b \in X$ (see also Remark \ref{remfu}). That is to say that all involutive set-theoretic $\check r$-matrices are coming from the permutation operator via the set-theoretic twist, i.e. $\check  r^F = F {\cal P} F^{-1}$ (and $r^F = F^{(op)} F^{-1}$) where recall ${\cal P}$ is the permutation operator and  $F =\sum_{a,b\in X} e_{a,a} \otimes e_{\sigma_{a(b), b}},$ such that for all $a,b,c \in X,$ $\sigma_a(\sigma_b(c)) =\sigma_{\sigma_a(b)}(\sigma_{\tau_b(a)}(c))$ and $\sigma_{\sigma_a(b)}(\tau_b(a)) = a$ (see also Corollary \ref{twist22} and \cite{DoRySt, DoikouYang}). 
 \end{rem}

\section{Solutions from Drinfel'd twists}
\noindent The main aim of this section is the derivation of general invertible solutions of the set-theoretic Yang-Baxter equation via an admissible Drinfel'd twist.

From the analysis of Section 5 we conclude that any admissible set-theoretic twist satisfies two fundamental conditions, for all $a,b, c \in X$,
$\sigma_a(\sigma_b(c)) = \sigma_{\sigma_{a}(b)}(\sigma_{\tau_b(a)}(c)),$ and $\sigma_{\sigma_a(b)}(\tau_b(a)) = \sigma_a(b) \triangleright a$ and $(X, \triangleright)$ is a rack (see also Theorem \ref{twist2}).
It is thus convenient to introduce an alternative definition for the admissible set-theoretic twist at the fundamental representation (see also Remark \ref{remfu2b} and a similar definition first introduced in \cite{DoRySt}).
\begin{defn} \label{admi2}
Let $(X, \triangleright)$ be a rack and define $F :=
\sum_{x,y\in X} e_{x,x} \otimes  e_{\sigma_x(y), y},$ such that $\sigma_x:X \to X$ is bijective. 
$F$ is called an admissible set-theoretic Drinfel'd twist if for all $a,b, c \in X$ (see also Theorem \ref{twist2})
\begin{enumerate}[{(a)}]
\item $\sigma_a(\sigma_b(c)) = \sigma_{\sigma_{a}(b)}(\sigma_{\tau_b(a)}(c)).$
\item $ \sigma_{\sigma_a(b)}(\tau_b(a)) = \sigma_a(b) \triangleright a.$
\end{enumerate}
\end{defn}
Notice that in the involutive case the second condition above becomes $\sigma_{\sigma_a(b)}(\tau_b(a)) =  a$ for all $a,b \in X.$

\subsection{Involutive case}

\noindent We first focus on the systematic derivation of involutive, set-theoretic solutions of the braid equation by exploiting the existence of an admissible Drinfel'd twist.

The following useful proposition can be now formulated (see also \cite{Ru07, Ru19, DoRy23}).
\begin{pro} \label{p1} Let $(X, \circ, 1)$ be a group and let $\sigma_a, \tau_b: X \to X,$ such that for all $a,b \in X,$ $\sigma_a$ is a bijection, $a\circ b = \sigma_a(b) \circ \tau_{b}(a)$ and $\sigma_{\sigma_a(b)}(\tau_b(a)) =a.$ 
Moreover, define $+: X\times X \to X,$ such that $a + b := a\circ \sigma^{-1}_a(b)$ 
and assume that $+$ is associative and for all $a,b,c \in X,$ $a\circ (b+c) = a \circ b -a +a \circ c$.
Then for all $a,b,c\in X,$
\begin{enumerate}
    \item $(X,+,\circ)$ is a brace and $\sigma_a(b) = -a + a\circ b.$ 
    \item  $\sigma_a(\sigma_b(c)) = \sigma_{\sigma_a(b)}(\sigma_{\tau_{b}(a}(c)).$ 
\end{enumerate}
    \end{pro}
\begin{proof}

$ $

\begin{enumerate}    

\item For the proof of $(X,+)$ being a group we refer the interested reader to \cite{DoRy23} for a step by step constructive approach of the algebraic structure. Then due to distributivity condition $(X,+, \circ)$ is a skew brace (see the original works on (skew) braces \cite{Ru05, Ru07, Ru19, GuaVen}). 
    To prove that $(X,+, \circ)$ is a brace we need to show that $(X,+)$ is an abelian group. 
    
    Indeed, from condition $\sigma_{\sigma_a(b)}(\tau_b(a)) =a$ and the structure group condition 
$a\circ b = \sigma_a(b) \circ \tau_b(a)$ we obtain for all $a,b \in X,$
\begin{equation}
\sigma_{a}(b)^{-1} \circ a \circ b = \sigma_{\sigma_a(b)}^{-1}(a)\ \Rightarrow\ a \circ \sigma_a^{-1}(b) = b \circ \sigma_b^{-1}(a)\ \Rightarrow\ a+b = b+a, \label{1} \nonumber
\end{equation}  
i.e. $(X,+)$ is abelian. We used  in the last part of the proof above the definition $a+b : = a \circ \sigma_a^{-1}(b).$ 

Moreover, for all $a,b \in X$
    \begin{equation}
        a\circ \sigma^{-1}_a(b) =  a+b \ \Rightarrow \ \sigma_a^{-1}(b) = a^{-1}\circ (a +b) \nonumber
    \end{equation}
and 
\begin{eqnarray}
&& \sigma^{-1}_{a}(\sigma_a(b)) = b\ \Rightarrow\ \
a^{-1} \circ (a + \sigma_a(b)) = b\ \Rightarrow \ \sigma_a(b) = -a +a \circ b. \qquad \qquad \qquad
 \hfill \nonumber
\end{eqnarray}

\item Condition (2) is condition (\ref{C1}) and its proof is given in the proof of Proposition \ref{pp0}. 
\hfill \qedhere
\end{enumerate}
    \end{proof}
    
 Notice that in Proposition \ref{p1} we assume a non-standard  distributivity 
 condition as the usual one does not apply (see for more details \cite{DoRy23}). Indeed, 
 let us assume that the usual distributivity condition holds, then 
 \begin{equation}
 a \circ (1+0) =a\ \Rightarrow \   a \circ 1 + a\circ 0 = a \ \Rightarrow\  a\circ 0 = 0 \ 
 \Rightarrow\ a = a \circ 0^{-1} = 1, ~~~\forall a \in X. \nonumber
 \end{equation}
 The latter statement is false given that we consider sets with not just the unit element $1.$ 
In general, when a non-empty set $X$ is equipped with two groups operation $\circ$ and $+,$ then the non-standard distributivity condition  of Proposition \ref{p1} (or slight variations of it \cite{DoRy23}) applies.

We conclude from Proposition \ref{p1} that the maps $\sigma_{a}(b)$ and $\tau_b(a)$  
provide an involutive solution $\check r (a, b) = (\sigma_{a}(b),\tau_b(a))$ of the set-theoretic braid equation 
(Rump's solution, Proposition \ref{pp0}).

    \subsection{Non-involutive case}
    
\noindent We recall from Section 4 (see also \cite{DoRySt}) that  invertible, non-involutive, set-theoretic solutions, which are the main focus of this subsection, are constructed from the rack-quandle solutions via an admissible Drinfel'd twist.

The following proposition is useful in describing general set-theoretic solutions.
\begin{pro} \label{p2}
Let $\sigma_a, \tau_b: X \to X,$ such that for all $a,b\in X,$ $\sigma_a$ is a bijection and $\tau_{b}(a) = \sigma^{-1}_{\sigma_a(b)}(\sigma_a(b)\triangleright a).$  Moreover, let $(X, \circ,1)$ be a group, such that for all $a,b \in X,$  
$a\circ b = \sigma_a(b) \circ \tau_{b}(a)$ and define $\bullet: X\times X \to X,$ such that $a \bullet b := a\circ \sigma^{-1}_a(b)\circ \xi,$ $\xi \in X$ is a fixed element.
Then, for all $a,b \in X,$  
\begin{enumerate}
\item  (a) $a \bullet b = b\bullet (b \triangleright a)$\\ (b)  $a\bullet \sigma_a(b) =  a\circ b\circ \xi.$ 
    \item Assume that $(X,\bullet)$ is a group and set $a\bullet b =: a+b$ and $\xi =1.$
    \begin{enumerate} 
    \item Then, $~ b\triangleright a =-b+  a + b,$ where $-b$ is the inverse of $b$ in $(X,+)$ (conjugate quandle),
    and $\sigma_a(b) = -a + a\circ b.$
    \item If $(X, +, \circ)$ is a skew brace, then  $\sigma_a(\sigma_b(c)) = \sigma_{\sigma_a(b)}(\sigma_{\tau_{b}(a{)}}(c))$ {and $\sigma_a(b\triangleright c) = \sigma_a(b) \triangleright \sigma_a(c).$} 
    \end{enumerate}
    \end{enumerate}
    \end{pro}
    \begin{proof}

$ $
    \begin{enumerate}
        \item  (a) Using the definitions  $a\bullet b = a \circ \sigma^{-1}_a(b)\circ \xi$ and $b\triangleright a =\sigma_b(\tau_{\sigma_a^{-1}(b)}(a))$ 
       for all $a,b \in X,$ we compute
        \begin{equation}
        b \bullet (b \triangleright a) = b \bullet  \sigma_b(\tau_{\sigma_a^{-1}(b)}(a)) =    
        b \circ \sigma^{-1}_b( \sigma_b(\tau_{\sigma_a^{-1}(b)}(a)))\circ \xi = b \circ \tau_{\sigma_a^{-1}(b)}(a))\circ \xi .    \nonumber
        \end{equation}
   But due to the condition  $a\circ b = \sigma_a(b) \circ \tau_{b}(a)$ we conclude that 
   $a \bullet b = b\bullet (b \triangleright a).$

\noindent (b) From the definition of $a\bullet b$, and the fact that $\sigma_a$ is bijection:
\begin{equation}
a\bullet \sigma_a(b) = a\circ \sigma_a^{-1}(\sigma_a(b))\circ \xi  = a \circ b \circ \xi. \nonumber
\end{equation}

\item (a) The first part follows immediately from (1) (a). From (1) (b)
 we immediately conclude that $\sigma_a(b) = -a +a \circ b.$

\noindent (b) This is condition (\ref{C1}) and it is shown in the proof of Proposal \ref{pp0}.
\hfill \qedhere
\end{enumerate}
\end{proof}

Notice that the binary operation $\bullet$ such that it satisfies condition (1) (a) in Proposition \ref{p2} is not uniquely defined (see \cite{DoRySt}  for a more detailed discussion). However, based on the definition of the operation $\bullet$ given in Proposition \ref{p2} we provide below a classification of non-involutive set-theoretic solutions given a specific rack/quandle.

In what follows we assume the existence of the map $\sigma_a: X \to X$ being a bijection and $(X, \circ)$ 
is a group.
\begin{enumerate}
\item {\bf The conjugate quandle.} 
This case corresponds to Part (2) of Proposition \ref{p2}. Recall that
$\sigma_a$ satisfies condition (a) of Definition \ref{admi2} and provides a solution to the Yang-Baxter equation. 
We also confirm that condition (b) of Definition \ref{admi2} is equivalent to $ a\circ b = \sigma_a(b) \circ \tau_b(a).$ 
This corresponds to the Guarnieri-Vendramin solution \cite{GuaVen}.

\item {\bf The affine quandle.} 
We generalize the definition of an affine quandle as follows. Let $f: X \to X$ be a a bijection and  
\begin{equation} 
f(-a+b+c) = -f(a) + f(b) + f(c) \label{qaf}
\end{equation}
for all $a,b,c \in X.$ Define also $\triangleright: X \times X \to X,$ such that 
$b\triangleright a = -f(b) + f(a) +b$ then $(X, \triangleright)$ is called an affine quandle.  
Note that, due to condition (\ref{qaf}) it is shown that $(X, \triangleright)$ is a quandle.
We define for all $a,b \in X$ $a \bullet b := f(a) + b,$ and conclude that $f(a) + b = f(b) +  
(b \triangleright a).$ Then we obtain from the definition of $\sigma^{-1}_a$ of Proposition \ref{p2}, $\sigma^{-1}_a(b) = a^{-1}\circ (f(a) +b)\circ \xi^{-1} $ 
and consequently $\sigma_a(b) =-f(a) + a \circ  b \circ \xi.$  Note that $\sigma_a$ satisfies condition (a) of Definition \ref{admi2} if and only if 
\begin{equation}
    a\circ f(b) - a + f(a) \circ \xi^{-1} = \sigma_a(b)\circ f(\tau_b(a)) - \sigma_{a}(b) + f( \sigma_a(b))\circ \xi^{-1}. \label{cc2}
    \end{equation}
    For instance, let $(X,+,\circ)$ be a skew brace and consider $f(a) := a \circ z - z,$ $z \in X$ is a fixed element and $\xi =1,$ then conditions (\ref{qaf}) and (\ref{cc2}) are satisfied.
    We also confirm that condition (b) of Definition \ref{admi2} is equivalent to $a\circ b = \sigma_a(b) \circ \tau_b(a),$ (see also \cite{DoRy22, DoRySt}).

\item {\bf The core quandle.} Recall the core quandle.  Recall that $(X, +)$ is a group and we define for all $a,b \in X,$ $b \triangleright a  = b -a +b.$ 
We also define for all $a,b \in X,$ $a\bullet b := -a+b,$ then $-a+b = -b + (b\triangleright a).$ 
Then according to Proposition \ref{p2} for $\xi =1$, $\sigma_a(b) =a + a \circ b.$  
{We confirm that $\sigma_a$ satisfies condition (a) and (b) of Definition \ref{admi2}. Notice, in particular that $\sigma_a(\sigma_b(c)) = \sigma_{a\circ b}(c),$ for all $a, b,c \in X,$ if and only if $(X,+,\circ)$ is a brace, i.e. $(X,+)$ is abelian.}
\end{enumerate}
And with this we conclude our discussion on the derivation of generic solutions of the set-theoretic Yang-Baxter equation from admissible twists. A more exhaustive analysis of admissible twists and general set-theoretic solutions in accordance to Proposition \ref{p2} will be presented in a forthcoming publication.

\subsection*{Acknowledgments}
\noindent 

\noindent I am indebted to G. Papamikos for sharing Figure 2 on the 3D consistency condition. I am grateful for the support 
from the Matrix Institute, University of Melbourne, and in particular the program ``Mathematics and Physics of Integrability'', 
1-19 July 2024 in Creswick, 
where part of this work was completed. Funding from the EPSRC research grant EP/V008129/1 is a also acknowledged.

\end{document}